\documentclass[11pt]{article}
\usepackage[hmargin=1in,vmargin=1in]{geometry}
\usepackage{hchang}
\usepackage{thm-restate} 
\usepackage{cleveref}
\usepackage{comment}
\usepackage{tcolorbox}
\usepackage{soul} 

\nolinenumbers
\usepackage{makecell}

\newtheorem{theorem}{Theorem}[section]
\newtheorem{lemma}[theorem]{Lemma}
\newtheorem{claim}[theorem]{Claim}
\newtheorem{observation}[theorem]{Observation}

\newtheorem{corollary}[theorem]{Corollary}

\newtheorem{definition}[theorem]{Definition}
\newtheorem{question}[theorem]{Question}

\def\eps{\e}
\def\cC{\mathcal{C}}
\def\cP{\mathcal{P}}
\def\cS{\mathcal{S}}
\def\cH{\mathcal{H}}

\def\cT{\mathcal{T}}

\newcommand{\Sketch}{\operatorname{sketch}}

\def\dist{\delta}
\newcommand{\ddim}{\ensuremath{d}}
\newcommand{\lca}{\ensuremath{\mathrm{lca}}}
\newcommand{\apices}{\ensuremath{\mathrm{Apices}}}

\title{Light Tree Covers, Routing, and Path-Reporting Oracles \\via Spanning Tree Covers in Doubling Graphs}

\date{}
\author{%
Hsien-Chih Chang%
\thanks{Department of Computer Science, Dartmouth College. Email: {\tt hsien-chih.chang@dartmouth.edu}.} 
\and 
Jonathan Conroy%
\thanks{Department of Computer Science, Dartmouth College. Email: {\tt jonathan.conroy.gr@dartmouth.edu}}  
\and 
Hung Le%
\thanks{Manning CICS, UMass Amherst. Email: {\tt hungle@cs.umass.edu}}  
\and
Shay Solomon%
\thanks{Tel Aviv University. Email: {\tt shayso@tauex.tau.ac.il}}  
\and
Cuong Than%
\thanks{Manning CICS, UMass Amherst. Email: {\tt cthan@cs.umass.edu}}  
}

\begin{document}
\maketitle
\begin{abstract}
A \emph{$(1+\eps)$-stretch tree cover} of an edge-weighted $n$-vertex graph $G$ is a collection of trees, where every pair of vertices has a
$(1+\eps)$-stretch path in one of the trees. 
The celebrated \emph{Dumbbell Theorem} by Arya \etal~[STOC'95] states that any set of $n$ points in $d$-dimensional Euclidean space admits a $(1+\eps)$-stretch tree cover with a constant  number of trees, where the constant depends on $\eps$ and the dimension $d$.
This result was generalized for arbitrary doubling metrics by Bartal \etal~[ICALP'19].
While the total number of edges in the tree covers of Arya \etal\ and Bartal \etal\ is $O(n)$, all known tree cover constructions incur a total 
\emph{lightness} of $\Omega(\log n)$;
whether one can get a tree cover of 
constant lightness has remained a longstanding open question, even for 2-dimensional point sets.

In this work we resolve this fundamental question 
in the affirmative, 
as a direct corollary of a new construction of $(1+\eps)$-stretch \emph{spanning} tree cover for doubling graphs; in a spanning tree cover,  every tree may only use edges of the input graph rather than the corresponding metric. 
To the best of our knowledge, this is the first constant-stretch
spanning tree cover construction (let alone for 
$(1+\eps)$-stretch) with a constant number of trees, for any nontrivial family of graphs.

Concrete applications of our spanning tree cover include:

\begin{enumerate}
\item A $(1+\eps)$-stretch tree cover construction, where both the number of trees and lightness are bounded by $O(1)$, for doubling graphs. In doubling metrics, we can also bound the maximum degree of each vertex by $O(1)$ (which is impossible to achieve in doubling graphs).
\item A compact $(1+\eps)$-stretch routing scheme in the labeled model
for doubling graphs, 
which uses the asymptotically \emph{optimal} (up to the dependencies on $\eps$ and $d$) bound of $O(\log n)$ bits
on all the involved measures (label, header, and routing tables sizes). 
This is a significant improvement over the works of Chan \etal~[SODA'05], Abraham \etal~[ICDCS'06], Konjevod \etal~[SODA'07], where the local memory usage either depends on the aspect ratio of the graph or is $\Omega(\log^3 n)$.
\item The first \emph{path-reporting} distance oracle for doubling graphs achieving \emph{optimal bounds} for all important parameters: $O(n)$ space, $(1+\eps)$-stretch, and $O(1)$ query time for  constant $d$ and $\eps$.

\end{enumerate}
\end{abstract}
\newpage
\section{Introduction}

\paragraph{Spanners.~} 
One of the most well-studied graph sparsifiers is the \EMPH{$t$-spanner}: a subgraph of a given edge-weighted graph $G$ that preserves pairwise distances up to a factor of $t$, called \EMPH{the stretch}. 
As one would expect, a $t$-spanner is often sparse (having a small number of edges), or light (having small total edge weights).%
\footnote{The \EMPH{sparsity} of a spanner $H$ is the ratio $|E_H|/|V|$ while the \EMPH{lightness} of $H$ is the ratio $w(E_H)/w(\mathrm{MST})$; here $w(E_H)$ is the total weight of edges in $E(H)$ and $w(\mathrm{MST})$ is the total weight of edges in the minimum spanning tree of $G$.}
Sparse spanners are well understood: constructions for various classes of graphs and metric spaces developed decades ago~\cite{ADDJS93,Cla87,Kei88,KG92,GR08Soda} are either optimal or conditionally optimal~\cite{ADDJS93,LS22}.
However, understanding light spanners has proved to be much more challenging, and for good reasons. 

At the conceptual level, every edge contributes equally to sparsity, while edges of different weights contribute differently to  lightness; in the most extreme case, one edge could contribute as much as $\Omega(n)$ other edges to the lightness.  This seemingly innocuous reason explains why, for several classes of graphs and metrics, constructing a light spanner was a long-standing open problem~\cite{DHN93,DNS95,GS02,Got15,CW15,BLWN17}.

A long line of research on light spanners over the last three decades~\cite{CDNS92,ADDJS93,DHN93, GS02, Got15, CW15, ENS15, FS20, BLWN17,BLW19, LS23, LST23, BF24} has discovered surprising results and techniques. 
It was curiously observed over the years that in all well-studied classes of graphs and metrics, lightness is only larger than sparsity by a factor of $\e^{-1}$. 
Representative examples (not meant to be exhaustive) include:

\begin{table}[h!]\small
\centering
\smallskip
\def\arraystretch{1.7}
\begin{tabular}{c:cc:cc:cc}
  \multicolumn{1}{c}{Input} & \multicolumn{2}{c}{Stretch} & \multicolumn{2}{c}{Sparsity (Ref.)} & \multicolumn{2}{c}{Lightness (Ref.)} \\ 
  \hline
  $\mathbb{R}^{O(1)}$ & $1+\eps$ & & $\Theta(\eps^{-d+1})$ & \cite{Kei88} & $\Theta(\eps^{-d})$ & \cite{LS22} \\
  \hdashline
  planar graphs & $1+\eps$ & & $\Theta(1)$ & & $\Theta(\e^{-1})$ & \cite{ADDJS93} \\
  \hdashline
  general graphs & \makecell{$2k-1$ \\ $(2k-1)(1+\eps)$} & \makecell{(sparse) \\ (light)} & $\Theta(n^{1/k})$%
  \footnotemark & \cite{ADDJS93} & $O(\e^{-1} n^{1/k})$ & \cite{LS23,Bodwin24} 
\end{tabular}
\end{table}

\footnotetext{For general graphs, the sparsity bound $O(n^{1/k})$ is optimal  assuming the Erd\H{o}s' Girth Conjecture.}

\noindent Whether the lightness bound $O(\e^{-1} n^{1/k})$ is tight for general graphs remains an outstanding open question; however, the recent (conditional) lower bound $\Omega(\e^{-1/k} n^{1/k})$ of Bodwin and Flics~\cite{BF24} suggested that this bound is probably optimal. 
A recent technical highlight is the work of Le and Solomon~\cite{LS23}, who showed that these seemingly disparate observations are instances of the same ``unified theme'': under a mild technical condition, they showed that, up to a factor of $O(\log(1/\eps))$:
\begin{equation}
\label{eq:unified-light}
\mathrm{Lightness}_{t(1+\eps)} \approx 
\e^{-1}\cdot\mathrm{Sparsity}_{t}
\end{equation}
where $\mathrm{Sparsity}_{t}$ and $\mathrm{Lightness}_{t}$ are the sparsity and lightness for $t$-spanners of the input graph classes.%
\footnote{See \cite[{Theorem~1.10(2)}]{LS23} for a precise statement of their result. 
In the case when $t = 1+\e$ there is an additional $+\e^{-2}$ additive term.
For well-studied classes of graphs and metrics, $\mathrm{Sparsity}_t \geq 1/\eps^2$ when $t=1+\eps$, and thus the additive factor $+1/\eps^2$ is dominated by  $\e^{-1}\cdot\mathrm{Sparsity}_{t}ps$.
Interestingly, they gave an example of a class of graphs that has a $(1+\eps)$-spanner with $O(1)$ sparsity, but any $(1+\eps)$-spanner must have lightness $\Omega(1/\eps^2)$, implying that the $+1/\eps^2$ additive term is unavoidable. 
}

\paragraph{Tree covers.~} While a spanner is a sparsifier that guarantees compactness (via sparsity and lightness), a \EMPH{tree cover} \cite{AP92,AKP94} is a sparsifier that seeks both \emph{compactness} and \emph{structural simplicity}.   
Formally, a tree cover of a graph $G$ is a set of trees $\mathcal{T}$ where each tree $T$ in $\mathcal{T}$ contains all vertices of $G$, and for every pair of vertices $(u, v)$, $d_T(u, v) \geq d_G(u, v)$. 
A \EMPH{$t$-tree cover} guarantees that for every pair of vertices $(u, v)$, there exists a tree $T$ in $\mathcal{T}$ that preserves the distance between $u$ and $v$ up to a factor of $t$, namely, $d_T(u,v)\leq t\cdot d_G(u,v)$.

The number of trees in the cover measures the \emph{sparsity} of a tree cover; and the fact that the cover contains only trees---the simplest type of connected graphs---signifies the structural simplicity. 
The structural simplicity of tree covers makes them a powerful tool for a wide range of applications where spanners fall short~\cite{AP92, AKP94, ADM+95, GKR01, BFN22, CCL+23, CCL+24a, CCL+24b}. 
For example, in network routing, designing efficient routing schemes for a spanner could be nearly as complex as doing so for the original graph. In contrast, routing schemes for tree covers are significantly simpler, as they reduce to the problem of routing on trees, where optimal routing schemes are known. Distance oracle construction is another application of tree cover: one could get a compact distance oracle for the input graph or metric space by constructing a distance oracle for every tree in the cover, which reduces to constructing a lowest common ancestor (LCA) data structure. Indeed, the best routing schemes and distance oracles for several classes of graphs and metrics are obtained via tree covers~\cite{TZ01,ACEFN20,CCL+23, CCL+24a,CCL+24b}.

However, constructing a tree cover with a small number of trees is much more demanding than constructing a sparse spanner due to the structural constraints imposed by a tree cover. 
After decades of research, the basic question of determining the optimal sparsity of $t$-tree covers remains open for all well-studied classes of graphs and metrics. 
In the very basic setting of $O(1)$-dimensional Euclidean space $\mathbb{R}^d$, Arya, Das, Mount, Salowe, and Smid~\cite{ADDJS93} constructed a tree cover with $O(\eps^{-d}\log 1/\eps)$ sparsity (number of trees) three decades ago; this bound was improved just a year ago to $O(\eps^{-d+1}\log (1/\eps))$~\cite{CCL+24b}, which is still $\log(1/\eps)$ factor away from the lower bound $\Omega(\eps^{-d+1})$. 
For planar graphs, a sequence of works has steadily improved the sparsity of $(1+\eps)$-tree cover from $O(\sqrt{n})$~\cite{GKR01} (here $\eps = 0$) to $O((\log n/\eps)^2)$~\cite{BFN22}, and recently to $O(\eps^{-3}\log(1/\eps))$~\cite{CCL+24b}; however, determining the sparsity as a function of $\eps$ remains widely open. 
For minor-free graphs, the sparsity even depends \emph{exponentially} on $1/\eps$~\cite{CCL+24a}. 
Interestingly, the technique developed for constructing a sparse tree cover of planar and minor-free graphs was used to solve the long-standing Steiner Point Removal problem in planar and minor-free graphs \cite{CCL+24b}, underscoring both the inherent difficulty and the deep connections between constructing a sparse tree cover and other major problems. 
One emerging phenomenon is that the optimal sparsity of $(1+\eps)$-tree cover seems to coincide with the optimal sparsity of $(1+\eps)$-spanner; one can interpret this as getting \ul{structural simplicity for free}. 
This is supported by the results in Euclidean and doubling metrics.
Nevertheless, we are very far from formally establishing this phenomenon for other classes of graphs and metrics. 

Although significant progress has been made on constructing sparse tree covers, the research on \emph{light} tree cover has seen almost no non-trivial progress. 
Given a $t$-tree cover $\mathcal{T}$ of a graph $G$, we define its \EMPH{individual lightness} to be the ratio $\max_{T\in \mathcal{T}} w(T) / w(\mathrm{MST})$ and its \EMPH{collective lightness} to be the ratio $\sum_{T\in \mathcal{T}} w(T) / w(\mathrm{MST})$ where $w(T)$ is the total weight of the edges of the tree $T$. 
The individual lightness measures how light each tree is in the cover, while the collective lightness measures the total lightness of all the trees in the cover. 
Observe that if a tree cover has sparsity $s$ and individual lightness $\iota$, then the collective lightness is at most $s\cdot\iota$; therefore, our focus is on constructing a sparse tree cover that is individually light. 
The state of the art for constructing an (individually or collectively) light tree cover has been extremely scarce:
even in the most well-studied Euclidean plane, light tree cover is not known to exist, and there is no viable (even conjectural) technique for the construction. 
In contrast, there are several different techniques~\cite{ADDJS93,CDNS92,DHN93,LS23} to construct a light spanner in $\R^2$.

All previous techniques for constructing a sparse tree cover in low-dimensional Euclidean/doubling metrics can be loosely interpreted as constructing a constant number of different hierarchical trees, such as quadtrees or net trees, by shifting.%
\footnote{The earliest sparse tree cover by \cite{ADM+95} for $\mathbb{R}^d$ could also be interpreted this way.}
However, a hierarchical tree could have (individual) lightness $\Omega(\log n)$, even in the basic case of the uniform line metric (a set of evenly spaced points on the line $\mathbb{R}$), hence the best that existing techniques could provide is a tree cover with individual lightness $O(\log n)$ (which is indeed achievable~\cite{FGN24}). 
One fundamental question is: 
\begin{question}
\label{ques:light-cover} 
Can we construct a $(1 + \eps)$-tree cover for point sets in $\mathbb{R}^d$ of sparsity $O_{\eps, d}(1)$ and individual lightness $O_{\eps, d}(1)$, independent of the number of points? Here $O_{\eps, d}(1)$ hides a dependency on $d$ and $\eps$.
\end{question}
Similar to the sparsity case, a positive answer to \Cref{ques:light-cover} could also be interpreted as getting structural simplicity for free as the union of all trees in the tree cover gives a light spanner.

\subsection{Main contribution}

The main result of this paper is to answer \Cref{ques:light-cover} affirmatively. Furthermore, our result is stronger in two ways. First, it applies to the wider family of doubling metrics\footnote{A metric space has doubling dimension $\ddim$ if every ball of radius $r > 0$ in the metric can be covered by at most $2^{\ddim}$ balls of radius $\frac{r}{2}$.}, and as discussed below it applies even more broadly. Second, every tree has a small vertex degree, so every tree is as  ``compact'' as possible.  

\begin{theorem}\label{thm:main} Given a point set  $P$ in a metric of constant doubling dimension $d$ and any parameter $\eps \in (0,1)$, there exists a $(1+\eps)$-tree cover for $P$ with sparsity $\eps^{-\Tilde{O}(d)}$ and individual lightness $\eps^{-O(d)}$. Furthermore, every tree in the tree cover has maximum degree bounded by $\eps^{-O(d)}$.
\end{theorem}

The $\Tilde{O}(g)$ notation hides a polylog factor in $g$.  As discussed above, the construction of a light tree cover in \Cref{thm:main} seems to require a new and completely different technique from all those used in previous work.  This is the key technical contribution of our work---we construct a light tree cover from a tree cover \emph{where every tree in the cover is a spanning tree} of the input graph. 

\paragraph{Spanning tree cover.~} 
We say that a tree cover $\mathcal{T}$ is a \EMPH{spanning tree cover} of a graph $G$ if every tree in $\mathcal{T}$ is a spanning subgraph of $G$. 
(Clearly, the notion of spanning tree cover only applies to graphs.) 
The \EMPH{doubling dimension of a graph} $G$ is the doubling dimension of its shortest path metric. 
Our goal is to construct a spanning $(1+\eps)$-tree cover with small sparsity (small number of trees) for a graph $G$ with a constant doubling dimension $d$; this question was indeed asked explicitly by Abraham, Gavoille, Goldberg, and Malkhi~\cite{AGGM06} almost two decades ago in the context of network routing.

All existing techniques~\cite{CCL+24a,CCL+24b,CCL+23,BFN22} are effective for constructing $(1+\eps)$-tree cover for {\em metric spaces}. 
These techniques, when tailored to the graph setting, will produce trees that contain edges and vertices that do not belong to the input graph (called \EMPH{Steiner edges} and \EMPH{Steiner vertices}), which are forbidden in a spanning tree cover. 
A natural idea is to replace a Steiner edge with a corresponding shortest path between the two endpoints. 
The issue is that there could be multiple Steiner edges in a single tree, and doing multiple replacements may create cycles. 
The fact that replacing Steiner edges with shortest paths does not work is hardly surprising: this is the key issue whenever one tries to construct a spanning version of a non-spanning object. 
A prime example is designing metric embedding into non-spanning trees vs spanning trees: the well-known FRT embedding~\cite{FRT04} gives a stochastic embedding of arbitrary graphs and metric spaces into a non-spanning tree with (optimal) expected distortion $O(\log n)$, while if one insists on an embedding into a \emph{spanning} tree, the best known expected distortion is $O(\log n\log\log n)$~\cite{AN19} despite significant research efforts and innovative technical ideas aimed at closing the gap~\cite{EEST08,ABN08,KMP11, AN19}. 

Here, we construct a spanning tree cover for doubling graphs based on two ideas: (i) a family of hierarchical partitions with \emph{strong diameter} guarantee and (ii) a collection of shortest paths that serve as a distance sketch called a \emph{preservable set}. 
Idea (i) partly imitates the construction of non-spanning tree cover~\cite{BFN22}; the difference here is that our clusters in the hierarchies of partitions have strong diameter guarantee so that the tree cover we build has the option to take the edges of the input graph, even though naively doing so will blow up the stretch to $\Omega(\log n)$.  
Idea (ii) is the most innovative one: roughly speaking, when we recursively construct a subtree of a given cluster, the preservable set technique seeks to preserve \emph{two shortest paths} (one of which is outside the cluster) in the recursion. 
One path plays the role of the ``spine'' or ``highway'' in the subtree, and the other path preserves the distance between a dedicated pair of vertices whose distance scale is the same as the scale of the cluster.  The former path guarantees that the recursion does not blow up the diameter of a cluster when we replace the cluster with a recursively constructed spanning subtree, while the latter path guarantees the stretch $1+\eps$ for the dedicated pair. 
(See \Cref{subsec:overview} for a more in-depth technical overview.) 
A surprising fact, perhaps, is that preserving a single shortest path is not sufficient to get a stretch of $1+\eps$ while preserving three or more paths seems impossible because they could form cycles. 
That is, our technique hits a sweet spot where preserving two paths is doable and, at the same time, sufficient for our purpose. 
Our end result is the first spanning tree cover for doubling graphs with a small number of trees.

\begin{theorem}
\label{thm:spanning-main}
Let $G $ be an edge-weighted graph with a constant doubling dimension $d$. 
For any $\eps \in (0,1)$, we can construct a spanning $(1+\eps)$-tree cover for $G$ with $\eps^{-\tilde{O}(d)}$ trees.  
\end{theorem}

We now showcase three applications of \Cref{thm:spanning-main}. First, as a direct corollary, we give the first construction of a light tree cover in Euclidean and doubling metrics, thereby proving \Cref{thm:main}, as mentioned earlier in the introduction. 
Second, we resolve a longstanding open problem: designing a $(1+\eps)$-stretch routing scheme in doubling graphs that achieves the asymptotically \emph{optimal} (up to the dependencies on $\eps$ and $d$) bound of $O(\log n)$ bits
on all the involved measures (label, header, and routing table sizes). 
Third, we show how to obtain a path-reporting distance oracle in doubling graphs with the optimal space and query time.
We believe that our spanning tree cover construction from \Cref{thm:spanning-main} will lead to further applications beyond those mentioned in this paper.

\paragraph{Light tree cover.~}
We now prove \Cref{thm:main} using \Cref{thm:spanning-main}.
First, given a set of points $P$ in a metric of doubling dimension $d$, we construct a $(1+\eps)$-spanner $H$ for $P$ with lightness and maximum vertex degree both bounded by $\eps^{-O(d)}$~\cite{Got15,BLW19,LS23}. It is well-known that $H$ has doubling dimension $O(d)$~\cite[Proposition~3]{Talwar04}). 
By \Cref{thm:spanning-main}, we can construct a spanning  $(1+\eps)$-tree cover $\mathcal{T}$ for $H$ that has $\eps^{-\tilde{O}(d)}$ trees. 
Observe that for every $T$ in $\mathcal{T}$, the weight of $T$ is at most the weight of $H$, which is $\eps^{-O(d)} \cdot w(\mathrm{MST})$, implying that the individual lightness of $\mathcal{T}$ is $\eps^{-O(d)}$. 
The stretch of $T$ with respect to the pairwise distances in $P$ is $(1+\eps)^2 \leq 1 + 3\eps$, which can be reduced to $1+\eps$ by scaling.

Our \Cref{thm:main} gives a light tree cover for point sets in the Euclidean plane as a special case. 
The typical path taken is that one solves a problem in Euclidean spaces (using Euclidean geometry) and then generalizes the technique to doubling metrics by using the packing bound in place of geometric primitives.
From this point of view, our path to the solution is unusual:  we use techniques developed specifically for doubling graphs to solve a problem for point sets in Euclidean space. More precisely, we construct a Euclidean $(1+\eps)$-spanner, use the fact that the spanner has doubling dimension $O(d)$, and apply our spanning tree cover theorem for doubling graphs. Obtaining a direct proof using Euclidean geometry without going through the route of doubling graphs remains challenging.  Our technique also gives as a corollary the spanning version of \Cref{thm:main}: a light spanning tree cover for doubling graphs. We note that, unlike the metric case, doubling graphs might not have a spanning tree cover of bounded degree: a counterexample is the star graph with exponentially increasing edge weights.  

\begin{corollary}\label{cor:light-graph} Let $G = (V,E,w)$ be an edge-weighted graphs with constant doubling dimension $d$. For any parameter $\eps \in (0,1)$, we can construct a $(1+\eps)$-tree cover for $G$ with sparsity $\eps^{-\Tilde{O}(d)}$ and individual lightness $\eps^{-O(d)}$. 
\end{corollary}
\begin{proof}
    We construct a $(1+\eps)$-spanner $H$ for $G$ with lightness $\eps^{-O(d)}$ by constructing a light spanner $K$ for the shortest path metric of $G$ and then replace every edge of $K$ with the shortest path between its endpoints in $G$.  Finally, we construct a spanning tree cover for $H$.
\end{proof}

\paragraph{Compact routing.~} 
Consider a distributed network of vertices in which each vertex has an arbitrary \emph{label}.
A \EMPH{routing scheme} is a distributed algorithm for routing messages 
from any source to any destination vertex in such a network, given the destination's label.
The most basic trade-off  is between the {\em (local) memory} at a vertex, used primarily to store {\em its routing table}, and the {\em stretch} of the routing scheme, which is the maximum ratio over all pairs between the length of the routing path induced by the scheme and the distance between that pair.
Other important quality measures are the sizes of {\em labels} and {\em (packet) headers}---which are typically dominated by the local memory usage due to the routing tables, and the {\em routing decision time}; refer to \Cref{S:routing} for more details on these quality measures.  

In the {\em name-independent} model, the designer of the routing scheme has no control over the labels of the vertices, which may be chosen adversarially.
In the {\em name-dependent model} (or in {\em labeled routing schemes}), the designer of the routing scheme can preprocess the network to assign each vertex a unique label.  
The edges incident to each vertex are given \EMPH{port numbers}; in the \EMPH{designer-port} model, these port numbers can be chosen by the algorithm designer in preprocessing, whereas in the \EMPH{fixed-port} model, these port numbers are chosen adversarially.

A routing scheme is called \EMPH{compact} if the local memory usage is low, which typically means polylogarithmic in the network size.
There is an extensive and influential body of work on compact routing schemes in general graphs (see, e.g., \cite{van1987interval,frederickson1988designing,eilam1998compact,cowen2001compact,TZ01,AGMNT04,GS11,chechik2013compact}), as well as in restricted graph classes, such as trees,  power-law graphs, graphs excluding a fixed minor and unit disk graphs, to name a few (see, e.g.,  \cite{TZ01,FG01,fraigniaud2002space,brady2006compact,chen2009compact,chen2012compact,kuhn2003ad,yan2012compact,kaplan2018routing,bodlaender1997interval,abraham2005compact}); in particular, compact routing schemes have been well studied also in Euclidean and doubling metrics, as well as in doubling graphs---which is the focus of this work \cite{Talwar04,Slivkins05,CGMZ16,AGGM06,KRX06,KRXY07,KRX07,GR08Soda,KRX08,ACGP10,KRXZ11,CCL+24b}. 

In the {name-independent} model, there exist networks of constant doubling dimension for which any routing scheme with a local memory usage of  $o(n^{(\eps/60)^2})$ 
must incur
 a stretch of at least $9-\eps$, for any $\eps > 0$
\cite{KRX06}, and there are $(9+\eps)$-stretch routing schemes with a local memory usage of 
 $O_{\eps,d}(\log^3 n)$ bits
\cite{AGGM06,KRX06,KRX07}.
We are interested in the regime of lower stretch, and in particular stretch $1+\eps$, and  henceforth restrict the attention to the name-dependent model.

In the name-dependent model, there are $(1+\eps)$-stretch routing schemes with 
a local memory usage (in bits) of either  $O_{\eps,d}(\log \Phi \cdot \log n)$%
\footnote{We denote by $\Phi$ the \EMPH{aspect ratio} of the graph, that is, the ratio  between the maximum and minimum distances.}~\cite{CGMZ16,Slivkins05,AGGM06}
or  
$O_{\eps,d}(\log^3 n)$ \cite{AGGM06,KRX07};
some of these works achieve the same bounds for the local memory usage, but increasingly better bounds for the sizes of headers and labels or for the dependencies on $\eps$ and $d$; for the sake of brevity, we focus on the bottom-line bound of local memory usage, ignoring dependencies on $\eps$ and $d$.
We note that for doubling (and Euclidean) metrics, there are $(1+\eps)$-stretch routing schemes with the asymptotically optimal (ignoring dependencies on $\eps$ and $d$) local memory usage of $O_{\eps,d}(\log n)$  \cite{GR08Soda,CCL+24b} bits.

A longstanding question is whether one can narrow the gap for $(1+\eps)$-stretch routing schemes in the name-dependent model---between doubling metrics and doubling graphs. In this work we settle this question completely using our spanning tree cover in \Cref{thm:spanning-main}. Not only does our $(1+\eps)$-stretch routing scheme
achieve the asymptotically optimal (ignoring dependencies on $\eps$ and $d$) local memory usage of
$O_{\eps,d}(\log n)$ bits, it does so in the fixed-port model, and it also achieves a routing decision time of $O_{\eps,d}(1)$. Our routing scheme, summarized in the statement below, is presented in \Cref{{S:routing}}.

\begin{restatable}{theorem}{RoutThm}
\label{thm:routing}
    Let $G$ be a graph with doubling dimension $\ddim$. For any $\e \in (0,1)$, there is a $(1+\e)$-stretch labeled routing scheme in the fixed-port model where the sizes of labels, headers, and routing tables are $\e^{-O(\ddim)}\cdot \log n$ bits. The routing decision time is $\e^{-O(\ddim)}$.
\end{restatable}

Another important advantage of our routing scheme over previous ones is that, loosely speaking, 
we carry out the routing protocol on top of a single tree in our tree cover. 
Now, the problem of labeled routing on trees has been well studied and there are efficient labeled routing schemes on trees (see in particular \cite{TZ01,FG01}); moreover, even from a practical perspective, it appears to be inherently more effective to route on trees rather than on basically any other graph class. 
Equipped with our spanning tree cover in \Cref{thm:spanning-main}, our routing algorithm consists of two stages. First, we need to determine a ``correct'' tree to route on for the input source and destination vertices,
i.e., one that approximately preserves their distance.
We achieve this using an appropriate labeling scheme, described in \Cref{ssec:find-correct-tree}.
In the second stage we route on top of that tree.
Alas, in general trees, fixed-port routing requires a local memory usage of $\Theta(\log^2 n/ \log \log n)$ bits: this is tight if one demands exact routing (stretch $1$) \cite{fraigniaud2002space}, and nothing better is known even if one allows a stretch of $1+\e$.
We demonstrate that we can bypass this lower bound for the {\em trees that we care about} (which come from our spanning tree cover). To this end, we first need to preprocess the input graph $G$ by constructing a greedy $(1+\eps)$-spanner $G'$ for $G$ \cite{ADDJS93}, and then construct a spanning tree cover for $G'$. We show that the resulting trees have ``convenient''  properties that make them amenable to very compact routing schemes. This is the focus of \Cref{ssec:routing-spanning-tree}.

\paragraph{Path-reporting distance oracle.~} 
An \EMPH{approximate distance oracle with stretch $t$} for an edge-weighted graph $G$ is a compact data structure for answering distance queries: given any two vertices $u$ and $v$ in $V_G$, return a distance estimate $\Tilde{d}(u, v)$ such that $d_G(u, v) \leq \Tilde{d}(u, v) \leq t\cdot d_G(u, v)$. A distance oracle is \EMPH{path-reporting} if, in addition to the distance estimate $\Tilde{d}(u, v)$, the oracle returns a path $\tilde{P}(u,v)$ in $G$ of length $\Tilde{d}(u, v)$. 
The \EMPH{query time} of a path-reporting distance is $q$ if for every pair $(u, v)$, it outputs $\tilde{P}(u,v)$ in time $q + O(|\tilde{P}(u,v)|)$, where $|\tilde{P}(u,v)|$ is the number of edges in $\tilde{P}(u,v)$. 
(It is often the case that reporting $\tilde{P}(u,v)$ takes more time than computing the distance estimate $\Tilde{d}(u, v)$, and hence the query time definition focuses only on the time it takes to find $\tilde{P}(u,v)$.)

A distance oracle construction often exhibits a trade-off between the space, stretch, and query time. The \EMPH{space} of a distance oracle is the number of machine words needed to store the oracle. In general graphs, the seminal result by Thorup and Zwick~\cite{TZ05} gives the first path-reporting distance oracle (PRDO) with space $O(k \cdot n^{1 + 1/k})$, stretch $2k-1$, and query time $O(k)$ for any given positive integer parameter $k$. The query time then was improved to $O(\log k)$ by Wulff-Nilsen~\cite{WN13}, but the space remained $O(k \cdot n^{1 + 1/k})$. 
One important question from a practical point of view is to construct a PRDO with \emph{linear space} and \emph{$O(\log n)$ stretch and query time}. 
Regardless of the choice of stretch parameter $k$, the oracle by  Thorup and Zwick~\cite{TZ05} (and Wulff-Nilsen~\cite{WN13}) has space $\Omega(kn^{1+1/k}) = \Omega(n\log n)$. Various attempts have been made to answer this question~\cite{WN13, EP16, ENW16, ES23, NS24, CZ24}, leading to the recent result by Chechik and Zhang~\cite{CZ24} who constructed an oracle with $O(n)$ space, $O(\log n)$ stretch and $O(\log \log \log n)$ query time. 

In the linear space regime, it is not possible to obtain stretch below $O(\log n)$ for general graphs, even for non-path-reporting oracle~\cite{TZ05}. On the other hand, a practical distance oracle should ideally have stretch  $1+\eps$ for any given fixed $\eps\in (0,1)$. For this, one has to assume additional structures on the input graph. One popular assumption is that the input graph has a bounded doubling dimension~\cite{AGGM06,FM20,KSGKSS24}. 
In this setting, Filtser \cite{filtser2019strong} constructed a PRDO with stretch $O(d)$, space $O(n \cdot d \cdot \log\log \Phi)$ and query time $O(\log \log \Phi)$ where $\Phi$ is the aspect ratio.  
While the space of Filtser's oracle is super-linear, the distortion is better than $O(\log n)$ for constant doubling dimension. On the other hand, a \emph{non-path-reporting} distance oracle with \emph{linear space} and $1+\eps$ stretch for a constant $\eps\in (0,1)$ and constant dimension $d$ is known for a long time~\cite{HM06,BKG11} (it should be noted that Filtser's construction has a  polynomial dependence on the dimension, which is better than the exponential dependence in ~\cite{HM06,BKG11}; but as mentioned, such dependencies are not the main focus of this paper). Can one design a path-reporting distance oracle in doubling graphs matching the bounds known for non-path-reporting ones? In this work, we resolve this question completely using our spanning tree cover in \Cref{thm:spanning-main}: our oracle achieves the optimal stretch, space, and query time.

\begin{theorem}
\label{thm:oracle} 
Let $G$ be an edge-weighted graph with a constant doubling dimension $d$. 
For any given parameter $\eps \in (0,1)$, we can construct a path-reporting distance oracle for $G$ with $\eps^{-\tilde{O}(d)} n$ space, $1+\eps$ stretch, and $\eps^{-\tilde{O}(d)}$ query time.   For constant $d$ and $\eps$, our oracle has $O(n)$ space, $1+\eps$ stretch, and $O(1)$ query time.
\end{theorem}

\begin{proof}
   Let $\mathcal{T}$ be a spanning $(1+\eps)$-tree cover for $G$ in \Cref{thm:spanning-main}. Then for each tree $T$ in $\mathcal{T}$, we construct an \EMPH{exact} path-reporting distance oracle $\mathcal{D}_T$ that has $O(n)$ space, and $O(1)$ query time (e.g.~\cite{LW21}). 
   The distance returned by  $\mathcal{D}_T$ is exactly $d_T(u,v)$. 
   The distance oracle for $G$ contains all oracles $\{\mathcal{D}_T\}_{T\in \mathcal{T}}$.  To query the distance between $u$ and $v$, we query $d_T(u,v)$ in $O(1)$ time using $\mathcal{D}_T$ for every $T\in \mathcal{T}$ and return $\min_{T\in \mathcal{T}}d_T(u,v)$. To find a $(1+\eps)$-approximate path, we query the path from $u$ to $v$ in $T^*$  by calling $\mathcal{D}_{T^*}$ where $T^* \coloneqq \arg\min_{T\in \mathcal{T}}d_T(u,v)$.  Observe that the query time is $O(|\mathcal{T}|) = \eps^{-\tilde{O}(d)}$ and the total space is $O(n \cdot |\mathcal{T}|) = \eps^{-\tilde{O}(d)}n$. The stretch bound follows from the $(1+\eps)$ stretch of $\mathcal{T}$.
\end{proof}

\subsection{Organization}

In \Cref{subsec:overview}, we give a high-level overview of the spanning tree cover construction (proof of \Cref{thm:spanning-main}). The formal construction of the spanning tree cover and all the needed technical components will be presented in \Cref{sec:spanning-alg}, \Cref{S:preservable-construct}, and \Cref{sec:HPF-construct}. Finally, in \Cref{S:routing}, we construct our routing scheme (proof of  \Cref{thm:routing}).

\section{Spanning Tree Cover: An Overview}\label{subsec:overview}

As we briefly discussed above, we construct our spanning tree cover from (1) a \emph{strong-diameter hierarchical partition family}, and (2) a \emph{preservable set}. While most of our innovative ideas are for defining and constructing a preservable set, the strong-diameter hierarchical partition lays an important groundwork for these ideas, and therefore, we will discuss it first. The definition of a family of hierarchical partitions is somewhat technical, but once the right definition is set up, one can see its connection to tree covers. 

\begin{figure}[!htb]
    \centering
    \includegraphics[width=1.0\textwidth]{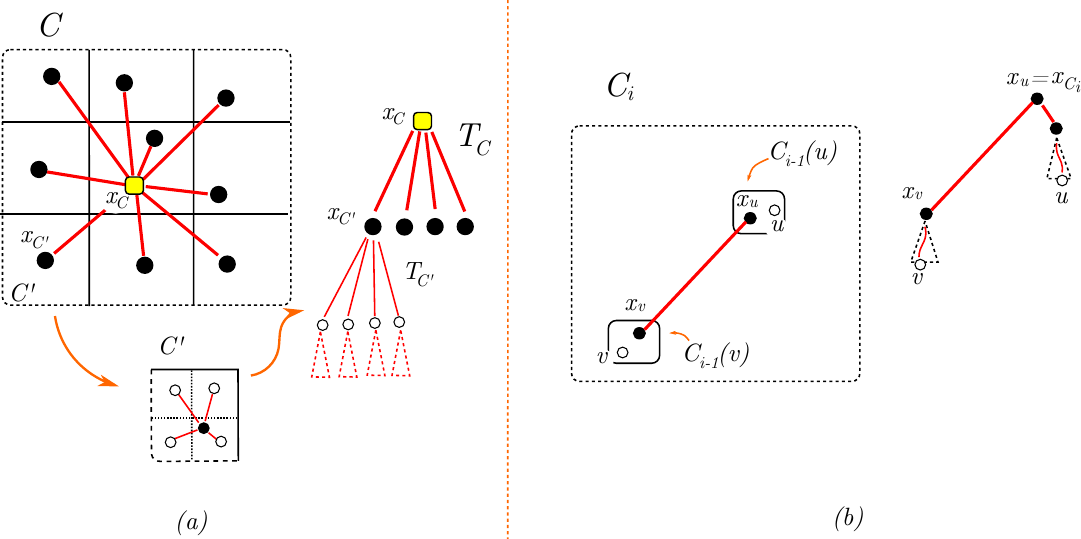}
    \caption{(a) A tree obtained from a hierarchy $\mathbb{P}^i \in \mathfrak{P}$ by connecting the representative $x_C$ of every cluster $C$ to those of its children. 
    In this example, $x_C$ has 9 children, but we only show 4 for a clearer picture. 
    (b) A $(1+\eps)$-stretch path from $u$ to $v$; clusters $C_{i-1}(u)$ and $C_{i-1}(v)$ have diameter roughly $\eps \cdot d_G(u,v)$.}
    \label{fig:hiearchy-to-tree}%
\end{figure}

\paragraph{Strong-diameter hierarchical partition.} A \EMPH{$\mu$-hierarchical partition}\footnote{Our definitions of the partitioning tools in this section (hierarchical partition and HPF) follow \cite{BFN22}, though similar ideas appeared earlier \cite{KLMN05}. The HPF is essentially a hierarchical version of sparse cover \cite{AP90}.} of $G$ is a sequence of partitions $\mathbb{P} = (P_0, P_1, \ldots, P_{i_\mathrm{max}})$ such that (1) every cluster in $P_0$ consists of a single vertex, and $P_{i_\mathrm{max}}$ consists of one cluster containing all vertices;
(2) for every $i$, every cluster $C$ in $P_i$ has (weak) diameter at most $\mu^i$ (that is, $d_G(x,y)\leq \mu^i$ for every $x,y\in C$; the distance is in graph $G$, not $G[C]$); (3) for every $i$, every cluster in $P_i$ is the \emph{union} of some clusters in $P_{i-1}$. One can view $\mathbb{P}$ as a tree where children of a cluster $C$ in the tree are the clusters at one level below whose union is $C$. 
A \EMPH{$(\rho, \mu)$-hierarchical partition family} (for short, a \EMPH{$(\rho, \mu)$-HPF}) is a set of  $\mu$-hierarchical partitions $\mathfrak{P} = \set{\mathbb{P}^1, \mathbb{P}^2, \ldots, \mathbb{P}^\sigma}$, such that for any vertex $v$ in $G$ and any scale $i \in \N$, there exists a hierarchical partition $\mathbb{P}^j \in \mathfrak{P}$ and a partition \smash{$P^j_i \in \mathbb{P}^j$} such that the ball $B(v, \mu^i/\rho)$\footnote{We use the notation $B(v,r)$ to denote the radius-$r$ ball centered at $v$.} is in some cluster $C$ in $P^j_i$.
The parameters $\sigma$, $\rho$, and $\mu$ are called the \EMPH{size}, the \EMPH{padding}, and the \EMPH{growth factor} of the HPF, respectively.

If one wants to construct a \emph{non-spanning} tree cover for doubling graphs (and metrics), then a $(\rho, \mu)$-hierarchical partition family  $\mathfrak{P}$ gives such a construction: for every hierarchy $\mathbb{P}^i \in \mathfrak{P}$, construct a tree $T_i$ by (a) choosing an appropriate vertex $x_C$ in each cluster $C$ as a representative of each cluster, and (b) connecting $x_C$ to the representative $x_{C'}$ by an edge if cluster $C'$ is the child of $C$ in $\mathbb{P}^i$; see \Cref{fig:hiearchy-to-tree}(a). 
Steiner edges of $T_i$ are those in step (b) since the edge $(x_C, x_{C'})$ may not be in $G$. The number of trees is $\sigma$.%
\footnote{In reality, the number of trees in the cover is the slightly larger $O(\sigma)$, due to a technical subtlety in choosing the representative $x_C$ in step (a), which we omit here for a simpler exposition.}
For the stretch of a pair $(u,v)$, the padding property of $\mathfrak{P}$ implies that there is a cluster $C_i$ of diameter $O(d_G(u,v))$ at some level $i$ of some partition $\mathbb{P}^j \in \mathfrak{P}$ that ``preserves'' the distance between $u$ and $v$.
More precisely, cluster $C_i$ that contains both $C_{i-1}(u)$ and  $C_{i-1}(v)$, which are the clusters at level $i-1$ containing $u$ and $v$ respectively. 
Moreover, a careful implementation of step (a) in doubling metrics (involving some duplication of the hierarchies in $\mathfrak{P}$) guarantees the existence of $\mathbb{P}^j$ and $C_i$ such that the representative $x_{C_i}$ of cluster $C_i$ satisfies  $x_{C_i} \in C_{i-1}(u) \cup C_{i-1}(v)$. Hence there is a direct edge between a vertex $x_u \in C_{i-1}(u)$ and a vertex $x_v \in C_{i-1}(v)$ (one of which is $x_{C_i}$, the representative of $C_i$); see \Cref{fig:hiearchy-to-tree}(b). 
By choosing the growth factor $\mu$ to be about $1/\eps$, both $d_G(u,x_u)$ and  $d_G(v,x_v)$ are around $\eps\cdot d_G(u,v)$ and hence the path $u$--$x_u$--$x_v$--$v$ in $T_i$ is a $(1+O(\eps))$-stretch path between $u$ and $v$. We note that existing non-spanning tree cover constructions~\cite{BFN22,CCL+24b} can also be reformulated as a construction from a hierarchical partition family. 

For the purpose of building a spanning tree cover, we cannot afford to include Steiner edges like $(x_C,x_{C'})$ in step (b); a necessary step is to guarantee the existence of a short path from $x_C$ to $x_{C'}$ in $G[C]$, the graph induced by vertices in $C$. 
This motivates us to construct a family of hierarchies with \emph{strong} diameter guarantee. 
We say that a $(\rho, \mu)$-hierarchical partition family  $\mathfrak{P}$ is of \EMPH{strong-diameter} if for every $j$, every cluster $C$ in $\mathbb{P}^j$ induces a subgraph $G[C]$ of diameter at most $\mu^i$. 
Known techniques~\cite{KLMN05,FL22B} construct a weak-diameter family of hierarchies by first overlaying clusters at level $i$ and clusters at level $i-1$, and then ``wiggling'' the boundary of clusters at level $i$ by assigning each cluster at level $i-1$ to an (arbitrary) intersecting cluster at level $i$ to get the hierarchical property. 
The wiggling process is problematic for strong diameter---even if we start with partitions of clusters with strong diameter for both level $i$ and $i-1$---since the wiggling effectively removes some vertices from one cluster at level $i$ and adds these vertices to another cluster at the same level. 
Removing vertices could blow up the strong diameter of the cluster, even to $+\infty$ (by disconnecting the cluster). 

We resolve this issue by observing that we can afford to assign a level-$(i-1)$ cluster $C_{i-1}$ to any level-$i$ cluster $C_i$ (even if $C_i$ does not intersect $C_{i-1}$) as long as the ``additive distortion'' is $O(\mu^{i-1})$, meaning that the distance from $C_{i-1}$ to $C_i$ is only larger than the distance from $C_{i-1}$ to its closest cluster at level $i$ by an additive amount $+O(\mu^{i-1})$. 
This flexibility enables us to use the \emph{cluster aggregation} technique recently introduced by Busch \etal~\cite{BDR+12,BCF+23} for solving the universal Steiner tree problem in doubling graphs to construct our strong-diameter family of hierarchies.

\begin{figure}[!htb]
    \centering
    \includegraphics[width=0.8\textwidth]{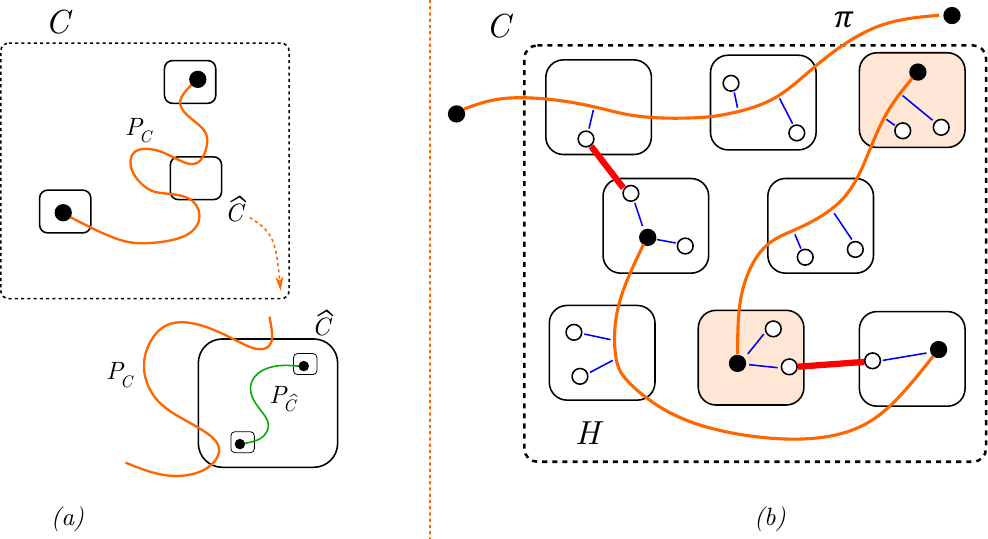}
    \caption{(a) The spanning tree $T_{\widehat C}$ for $\widehat C$ has to preserve both the subpath $\pi_C$ inside $\widehat C$, as well as the shortest path $\pi_{\widehat C}$ associated with $\widehat C$. 
    (b) A sketch graph $H$ with respect to preservable set $\cP$ (the orange paths; this includes the highway $\pi$ of $C$, and the path between the preserving pair of $C$ belonging to shaded clusters), inter-cluster edges $I$ (the bold red edges), and fake edges (the thin blue edges).}
    \label{fig:sketch-graph}%
\end{figure}

\paragraph{Preservable set and sketch graph.}  
In step (b) above, for \emph{every child} cluster $C'$ of $C$,  we connect $x_C$  to $x_{C'}$ by a Steiner edge.  
Since $C$ has a strong diameter, a simple idea to avoid Steiner edges is to connect $x_C$ to every child representative $x_{C'}$ by a shortest path in $G[C]$. 
The issue with this idea is that when we recursively process a child cluster  $C'$ of $C$, we have to add more shortest paths inside $G[C']$, and these paths may interact with the shortest paths in $G[C]$ and form cycles. 
We tackle this problem by further imposing an additional \EMPH{pair-preserving property} on the family $\mathfrak{P}$ (see \Cref{def:pair_HPF}):
\begin{itemize}
\item (P1) It suffices to connect $x_C$ to a \emph{single} child representative $x_{C'}$ by a shortest path, denoted~$P_C$. 
We call $P_C$ the \EMPH{shortest path associated with $C$} and the pair $(x_C, x_{C'})$ a \EMPH{preserving pair} of $C$. 
\item (P2) We can relax the distance guarantee of other shortest paths (from $x_C$ to other child representatives) to $O(\mu^i)$, where $\mu^i$ is the upper bound on the diameter of $C$. 
\end{itemize}
A more intuitive way to interpret (P1) and (P2) is that we have to construct a spanning tree $T_C$ for $C$ of diameter $O(\mu^i)$ containing the shortest path $P_C$. 
Other than its endpoint clusters, path $P_C$ might intersect other child cluster $\hat{C}$ of $C$; hence when we construct the spanning tree $T_{\hat{C}}$ for $\hat{C}$, it has to preserve \emph{both} paths: the subpaths $P_C$ inside $\hat{C}$, and its own shortest path $P_{\hat{C}}$ (see Figure~\ref{fig:sketch-graph}(a)). 
This is the starting point of our idea of a \emph{preservable set} and its associated \emph{sketch graph}, which we will introduce next. 

\smallskip
Let $C$ be a cluster in a hierarchy with diameter $\mu^i$, and let $\mathcal{C}$ denote the set of child clusters of $C$ in the same hierarchy. 
Clusters in  $\mathcal{C}$ have a diameter of roughly $\eps \mu^i$.  
Let $\pi$ be a shortest path in $G$ such that $\pi\cap G[C]\not=\emptyset$ and $\pi$ might contain vertices not in $G[C]$; the path $\pi$ is associated with some ancestral cluster of $C$ in (P1) and is called the \EMPH{highway} of $C$. 
A set of \emph{vertex-disjoint paths} $\mathcal{P}$ in $G[C] \cup \pi$ is a \EMPH{preservable set} with respect to $(G[C], \mathcal{C}, \pi)$ if (i) $\mathcal{P}$ contains $\pi$ and (ii) every cluster $C'\in \cC$ is touched by \emph{exactly one path $\pi_{C'}$} in the preservable set $\mathcal{P}$. 

The definition of the preservable set directly encodes (P1). 
To resolve (P2), the preservable set  $\mathcal{P}$ has to satisfy another condition formulated via the \emph{sketch graph}; see \Cref{fig:sketch-graph}(b). Roughly speaking, the \EMPH{sketch graph} has $C \cup \pi$ as a vertex set. The edge set contains three types of edges: 
(a) all edges of paths in  $\mathcal{P}$, (b) for every cluster $C'$ in $\mathcal{C}$, a \ul{fake edge} between every vertex $v\in C'$ and the closest vertex in $\pi_{C'} \cap C$ of weight $O(\eps \mu^i)$---here $\pi_{C'}$ is the unique path in $\cP$ touching $C$---and (c) some edges between clusters in $\mathcal{C}$, called \EMPH{inter-cluster edges}. See the formal definition in \Cref{def:sketch}. 
The idea of fake edges in (b) is that when we recursively construct a spanning tree $T_{C'}$ of $C'$ (passing down the path $\pi_{C'}$ as its highway) containing edges of $\pi_{C'}$ in $C'$, we can remove all the fake edges in (b) and add in the edges of $T_{C'}\setminus \pi_{C'}$; the weight bound $O(\eps \mu^i)$ of the fake edges are the diameter upper bound of  $T_{C'}$, which we can guarantee recursively. 

Our key technical contribution is the preservable set lemma (\Cref{lm:1+e_distort_preserve_set}) showing that we can construct a sketch graph that (i) is a tree, (ii) preserves the distance between the preserving pair of $C$ in (P1), and (iii) has diameter $O(\mu^i)$. 
Properties (ii) and (iii) resolve (P1) and (P2) mentioned above. Note that the sketch graph is not a spanning tree since it contains fake edges; as pointed out above, we can convert it to a spanning tree by replacing fake edges inside a cluster with edges of a spanning tree constructed recursively for that cluster. 

\paragraph{Spanning tree cover construction.} 
Once equipped with the two objects, the spanning tree cover construction can be summarized as follows.
Given the input graph $G$,
construct a strong-diameter, pair-preserving $(O(1),\mu)$-HPF $\mathfrak{P}$ of size $O_{\e, \ddim}(1)$  such that each cluster is associated with a preserving pair in (P1).  For each hierarchy $\mathbb{P} \in \mathfrak{P}$, construct a tree $T_{\mathbb{P}}$ by calling the recursive procedure
\textsc{PathPreservingTree}$(\mathbb{P}, G, \varnothing)$;
then return $\mathcal{T} \coloneqq \{T_{\mathbb{P}}: \mathbb{P} \in \mathfrak{P}\}$.
\textsc{PathPreservingTree}$(\mathbb{P}, G[C], \pi)$ takes three arguments: a $\mu$-hierarchical partition $\mathbb{P}$, a cluster $C$ in the hierarchy, and a path $\pi$. 
It returns a spanning tree of $G[C] \cup \pi$ using the following steps.
\begin{enumerate}
    \item  
    Let $\mathcal{C}$ be the set of children of $C$ in the hierarchy $\mathbb{P}$.
    \item  
    Let $\mathcal{P}$ be a preservable set with respect to $(G[C], \mathcal{C}, \pi)$, and $H$ be its sketch graph.
    \item  
    For each cluster $C' \in \mathcal{C}$, let $T_{C'} \gets$ \textsc{PathPreservingTree}$(\mathbb{P}, C, \pi_{C'})$. 
    Recall that $\pi_{C'}$ is the (unique) path in the preservable set $\mathcal{P}$ touching $C'$.
    \item   
    Return $\Paren{\bigcup_{C' \in \mathcal{C}} T_{C'}}$ connecting by the set of inter-cluster edges in $H$. 
\end{enumerate}
The number of trees in our cover is exactly the number of hierarchies in $\mathfrak{P}$, which is $O_{\ddim,\eps}(1)$. All the trees are spanning since we replaced fake edges of $H$ with edges of $T_{C'}$ in step 3, which are edges of $G$.  The stretch is $(1+\eps)$ since the sketch graph preserves the distances between the preserving pair.

\section{Spanning Tree Cover Algorithm} \label{sec:spanning-alg}

In this section, we make precise the high-level structure of the $(1+\e)$-stretch spanning tree cover algorithm. 
We first introduce the two key tools, whose proofs we defer to Sections~\ref{S:preservable-construct} and \ref{sec:HPF-construct}.
For the purpose of constructing spanning tree covers, we can safely assume without loss of generality that all graphs are connected.

\subsection{Strong-diameter pair-preserving HPF}

\begin{definition}
\label{def:HP}
    A (strong-diameter) \EMPH{$\mu$-hierarchical partition} of $G$ (for short, \EMPH{$\mu$-HP}) is a sequence of partitions $\mathbb{P} = (P_0, P_1, \ldots, P_{i_\mathrm{max}})$:
    \begin{itemize}
        \item Every cluster in $P_0$ consists of a single vertex, and $P_{i_\mathrm{max}}$ consists of one cluster containing all vertices.
        \item For every $i \in [i_\mathrm{max}]$, every cluster $C$ in $P_i$ has (strong) diameter at most $\mu^i$; that is, 
        \(
        \max_{u,v \in C} \dist_{G[C]}(u,v) \le \mu^i.
        \)
        \item For every $i \in [i_\mathrm{max}]$, every cluster in $P_i$ is the union of some clusters in $P_{i-1}$.
    \end{itemize}
    The partition $P_i$ is called a \EMPH{scale-$i$} partition.
\end{definition}

\begin{definition}
    A (strong-diameter) \EMPH{$(\rho, \Delta$)-sparse partition family} is a set of partitions $\set{P^1, \ldots, P^\sigma}$ of the vertices of $G$, where each $P^j$ is a partition of $G$ into clusters such that:
    \begin{itemize}
        \item Every cluster $C$ has (strong) diameter at most $\Delta$.
        \item For every vertex $v$, there is a partition $P^{j}$ and a cluster $C$ in~$P^{j}$ so that the ball of radius $\Delta/\rho$ centered at $v$ is contained in~$C$.
    \end{itemize}
    The parameter $\sigma$ is called the \EMPH{size}, and the parameter $\rho$ is the \EMPH{padding}.
    We say that $G$ admits a \EMPH{$\rho$-sparse partition family} (\EMPH{$\rho$-PF}) of size $\sigma$ if, for every $\Delta$, it admits a $(\rho, \Delta)$-sparse partition family of size~$\sigma$.
\end{definition}

The next object \emph{simultaneously} achieves the ``hierarchical'' property of a $\mu$-hierarchical partition and the ``padding'' property of a $\rho$-sparse partition family.

\begin{restatable}{definition}{HPFdef}
    A (strong-diameter) \EMPH{$(\rho, \mu)$-hierarchical partition family} (for short, a \EMPH{$(\rho, \mu)$-HPF}) is a set of (strong-diameter) $\mu$-hierarchical partition $\mathfrak{P} = \set{\mathbb{P}^1, \ldots, \mathbb{P}^\sigma}$, such that for every vertex $v$ in $G$ and every $i \in \mathbb{N}$,
    there exists some $\mu$-hierarchical partition $\mathbb{P}^j \in \mathfrak{P}$ whose scale-$i$ partition $\smash{P^j_i \in \mathbb{P}^j}$ contains some cluster $C \in \smash{P^j_i}$ such that the ball of radius $\mu^i/\rho$ centered at $v$ is contained in $C$.
\end{restatable}
We view each hierarchical partition $\mathbb{P} = (P_1, P_2, \ldots P_{i_{\max}})$ as a tree structure, in which each cluster $C \in P_i$ is a node at level $i$ and the clusters in $P_{i - 1}$ whose union is $C$ are the children of $C$. We use the same notation for \EMPH{parent}, \EMPH{ancestor}, \EMPH{descendants}, and \EMPH{degree} for each cluster in $\mathbb{P}$ as in a tree.
We remark that an HPF can either be viewed as a set of hierarchical partitions with an additional covering property (as in the definition above), or as a set of sparse partition families with additional hierarchical structure.

\medskip
We will construct an $(O(1), \tilde{O}(\ddim^3))$-HPF of size $2^{O(\ddim)}$ with \emph{strong diameter} guarantee for graphs with bounded doubling dimension in Section~\ref{sec:HPF-construct}.
One benefit of working with doubling graphs is that we can additionally assume the number of child subclusters is constant for any given cluster.
The \EMPH{degree} of a cluster in a hierarchical partition $\mathbb{P}$ is the total number of its children.

\begin{restatable}{theorem}{HPF}
\label{lem:HPF}
    If a graph $G = (V, E)$ has doubling dimension $\ddim$, then it admits a strong-diameter $(\rho, \mu)$-HPF $\mathfrak{P}$ of size $\smash{2^{\tilde{O}(\ddim)}}$ with $\rho = O(1)$ and $\mu = \tilde{O}(\ddim^3)$, where the maximum degree of any cluster in any HP in $\mathfrak{P}$ is at most $\mu^{O(\ddim)}$.
\end{restatable}

To motivate the following definition, we draw analogy to the quadtree construction~\cite{SA12,ACRX22} in the Euclidean setting.
Let $Q$ be a cell in $\R^2$ with side-length $D$ containing every point in consideration, with minimum distance $1$ between every pair of points.  
We construct a quadtree on $Q$ by making $Q$ to be the root of the tree at level $\log_2 D$, recursively construct a quadtree for each of the four subcells of $Q$ with side-length $D/2$, and append the roots of the resulting four quadtrees to the root~$Q$.
We stop the process when the cell has side length $1$.
For each cell $C$ in this quadtree, the subcells contained in $C$ with side-length $\e D$ all lie at $\Ceil{\log_2 (1/\e)}$ levels down below $C$ in the quadtree.
We want to give special attention to these subcells that are $\e$-fraction smaller because, for the purpose of distance estimation, things that happen within these subcells could be considered negligible. 

Going back to HPFs, for each cluster $C$ in a $\mu$-hierarchical partition $\mathbb{P}$, we define \EMPH{$\e$-subclusters} to be the clusters that are $\ceil{\log_\mu (1/\e)}$ levels below $C$.  (The value $\EMPH{$\ell$} \coloneq \ceil{\log_\mu (1/\e)}$ is called the \EMPH{offset}.)
By the definition of a $\mu$-HP\/, any $\e$-subcluster of a level-$i$ cluster $C$ must be at level $i-\ell$ and have (strong) diameter at most $\e \cdot \mu^{i}$.
(For $i$-th level clusters where $i<\ell$ we assume their $\e$-subclusters are the ones at level $0$, which by definition are singleton points and have diameter $0$.)
A \EMPH{($\rho$-well-separated) subcluster pair} of a cluster $C$ is a pair of $\e$-subclusters $(C_1,C_2)$ of $C$ such that 
$d_{G} (C_1, C_2) > \mu^{i} / \rho$.%
\footnote{We can achieve better bounds by relaxing the lower bound guarantee; however this leads to only changes in dependency on $\ddim$, which we choose not to optimize for better clarity in presentation.}
(When $\rho$ is clear from context, we sometimes omit it and just call $(C_1,C_2)$ a \emph{subcluster pair}.)
This is analogous to the WSPD construction in the Euclidean setting~\cite{CK95}; intuitively, it is sufficient to preserve distance between two subclusters without worrying about the actual representatives from each subclusters, again because things that happen within the $\e$-subclusters are negligible.

\begin{definition}[Pair-Preserving HPF]
\label{def:pair_HPF}
We say a $(\rho, \mu)$-HPF $\mathfrak{P}$ of a graph $G$ is \EMPH{pair-preserving} if
every cluster $C$ is associated with a ($\rho$-well-separated) subcluster pair $(C_1,C_2)$, such that
for any pair of vertices 
$u$ and $v$ in $G$, there exists an HP $\mathbb{P} \in \mathfrak{P}$ containing some cluster $C$ such that:
\begin{enumerate}
    \item \label{it:assign_HPF} 
    letting $(C_1,C_2)$ denote the subcluster pair associated with $C$, we have $u \in C_1$ and $v \in C_2$;
    \item \label{it:dist_prev_HPF} 
    $d_{G}(x, y) = d_{G[C]}(x, y)$ for any $x \in C_1$ and $y \in C_2$.
\end{enumerate}
We say that the HP $\mathbb{P}$ \EMPH{preserves} the pair $(u,v)$.
\end{definition}

It is not hard to turn our $(\rho, \mu)$-HPF construction from \Cref{lem:HPF} pair-preserving; we just have to make sufficiently many copies of each cluster $C$, so that there is one copy dedicated to every subcluster pair of $C$. 
Due to the packing bound, it suffices to make $\eps^{-O(\ddim)}$ copies.

\begin{lemma}
\label{lm:pair_preserve_HPF}
    Given a graph $G$ of doubling dimension $\ddim$. 
    For any sufficiently small constant $\e$ such that $\e^{-1} \ge \tilde{\Omega}(\ddim^3)$, one can construct a pair-preserving strong-diameter $(\rho, \mu)$-HPF $\mathfrak{P}$ on $G$ of size $\e^{-O(\ddim)}$  with $\rho = \tilde{O}(\ddim^3)$ and $\mu = \tilde{O}(\ddim^3)$.
\end{lemma}

\begin{proof}
From \Cref{lem:HPF}, $G$~admits a strong-diameter $(\rho', \mu)$-HPF $\mathfrak{P}'$ of size $\smash{2^{\tilde{O}(\ddim)}}$ with $\rho' = O(1)$, $\mu = \tilde{O}(\ddim^3)$ and maximum degree $\mu^{O(\ddim)}$.
We choose $\rho \coloneqq D \cdot \mu \rho'$ for some constant $D > 0$ to be determined later.
For each HP $\mathbb{P}'$, we create $k$ copies of $\mathbb{P}'$ for some number $k = \e^{-O(\ddim)}$ to be determined, and add them to a new HPF\/, called $\mathfrak{P}$.  We then assign each cluster $C$ in each copy to a $\rho$-well-separated subcluster pair of $C$. 
Specifically, for each cluster $C$ at level $i$ of $\mathbb{P}'$, let $C_1, C_2, \ldots, C_k$ be the copies of $C$ in the copies of $\mathbb{P}'$. 
Because every cluster $C$ has degree $\mu^{O(\ddim)}$, $C$ has at most $\mu^{O(\ddim) \cdot \log_{\mu}(1/\e)} = \e^{-O(\ddim)}$ descendants at level $i - \ell$; hence, 
there are at most $\e^{-O(\ddim)}$ subcluster pairs of $C$.
Thus, we can set $k$ so that there are more copies of $C$ than the number of subcluster pairs. 
We then associate each copy of $C$ to a unique ($\rho$-well-separated) subcluster pair. 
(If there are more copies than the subcluster pairs, we assign the remaining copies to an empty pair.)

\smallskip
We now show that $\mathfrak{P}$ satisfies both conditions in \Cref{def:pair_HPF}. 
Letting $i$ be the smallest integer such that $d_G(u, v) \leq \frac{\mu^i}{2\rho'}$; by minimality of $i$, we have
\(
    d_G(u, v) > \smash{\frac{\mu^{i - 1}}{2\rho'}}.
\)
By the definition of $\mathfrak{P}'$, there is a hierarchy $\mathbb{P}'$ such that $B(u, \mu^i/\rho')$
is in some cluster at level $i$ of $\mathbb{P}'$, called $C$.
Note in particular that vertices $u$ and $v$ are both in cluster $C$.
Let $C_u$ and $C_v$ be the two $\e$-subclusters containing $u$ and $v$, respectively, at level $i - \ell$.
For \cref{it:assign_HPF},
because $d_G(u, v) > \frac{\mu^{i - 1}}{2\rho'} = \frac{D \mu^{i}}{2\rho}$ and the subclusters have diameter at most $\e\mu^i$, the two subclusters $C_u$ and $C_v$ are $\rho$-well-separated if we set 
$D \ge 4$ and $\e \le \frac{1}{8\mu\rho'}$.
Then, by our construction of $\mathfrak{P}$, there exists a copy of $\mathbb{P}'$ in which we assign $(C_u, C_v)$ to (a copy of) $C$; call this copy $\mathbb{P}$. 
Thus, $\mathbb{P}$ satisfies \cref{it:assign_HPF}.

As for \cref{it:dist_prev_HPF},
because the subclusters have strong diameter $\e\mu^i$,
if we pick two vertices $x \in C_u$ and $y \in C_v$, we have both $d_G(u,x)$ and $d_G(v,y)$ at most $\e\mu^i$.
Since $d_G(u, v) \leq \frac{\mu^i}{2\rho'}$, 
every vertex in the shortest path from $x$ to $y$ lies entirely in 
$B(u, \mu^i/2\rho' + 2\e\mu^i) \subseteq B(u, \mu^i/\rho')$ as $\e = O(\ddim^{-3}) \le 1/4\rho'$, and thus also in $C$, yielding \cref{it:dist_prev_HPF}:  
\(
d_{G[C]}(x, y) = d_G(x, y)
\)
for any $x \in C_u$ and $y \in C_v$.
$\mathbb{P}$ inherits \cref{it:dist_prev_HPF} from~$\mathbb{P}'$.  
\end{proof}

\subsection{Preservable sets}

Let \EMPH{$\widehat{G}$} be a subgraph with diameter \EMPH{$D$} of a graph $G$.  
Let \EMPH{$\mathcal{C}$} be a clustering of $\widehat{G}$ into clusters of strong diameter $O(\e D)$. Let \EMPH{$\pi$} be a path in $G$ \emph{not necessarily contained in $\widehat{G}$}, such that $\pi$ is a shortest path in the graph $\widehat{G} \cup \pi$.
We emphasize again that $\pi$ may not be in the subgraph $\widehat{G}$, which is why we consider the graph $\widehat G \cup \pi$ in the definitions below.

\begin{definition}[Preservable set]
\label{def:preservable} 
Fix a graph $\widehat{G}$, clustering $\mathcal{C}$, and path $\pi$. Let $\mathcal{P}$ be a set of paths in the graph $\widehat{G} \cup \pi$. 
We say that $\mathcal{P}$ is a \EMPH{preservable set} with respect to $(\widehat{G}, \mathcal{C}, \pi)$ if: 
\begin{enumerate}
    \item \label{it:cluster_disjoint} 
    Every cluster is \EMPH{touched} by \emph{exactly one} path in $\mathcal{P}$.
    That is, for every cluster $C\in \cC$, there exists a unique path \EMPH{$\pi_C$} in $\cP$ such that $C\cap V(\pi_C)\not=\varnothing$. 
    In particular, this means that the collections of subsets of clusters touched by every path in $\cP$ are pairwise disjoint, and paths in $\mathcal{P}$ are vertex-disjoint paths. 
    \item \label{it:pi} $\pi$ itself is in the preservable set $\mathcal{P}$.
    \item \label{it:shortest_path} 
    For every $C \in \cC$, $\pi_C$ is a shortest path in $\widehat G[C] \cup \pi_C$. 
\end{enumerate}
\end{definition}

We call $\cP$ a preservable set because, when we construct a tree $T$ for our tree cover (in Section~\ref{SS:tree-cover}) via a recursive procedure, we will be able to guarantee that every path in $\cP$ is in the tree $T$ --- that is, the tree $T$ \emph{preserves} these paths. We now give a definition and lemma which let us choose a ``good'' set of paths to preserve --- that is, a preservable set which helps to satisfy the distortion guarantee of tree cover.

\begin{definition}[Sketch graph]
\label{def:sketch}  
Let $\cP$ be a preservable set with respect to $(\widehat{G},\cC,\pi)$, 
and let $I$ be a set of \EMPH{inter-cluster edges} from $\widehat{G}$ with respect to $\cC$ (that is, no edge has both endpoints in the same cluster of $\cC$).
The sketch graph \EMPH{$\Sketch(\mathcal{P}, I)$} is a graph with vertex set $V(\widehat{G} \cup \pi)$ and edge set that is comprised of the following:
\begin{enumerate}
    \item  Take all edges in all paths of $\mathcal{P}$.
    \item  For every cluster $C$ in $\mathcal{C}$, 
    for every vertex $v$ in $C$, add a \emph{fake edge} between $v$ and the closest vertex in $\pi_C \cap C$ with weight $10 \cdot \e D$.
    (Recall $\pi_C$ is the unique path in $\cP$ touching $C$, if exists.)
    \item  Add the inter-cluster edges $I$ to form the final sketch graph $\Sketch(\mathcal{P}, I)$.
\end{enumerate}  
If there are no cycles in $\Sketch(\cP, I)$, we say the sketch graph and the inter-cluster edges $I$ are \EMPH{cycle-free}.
\end{definition}

We observe that a cycle-free sketch graph $\Sketch(\mathcal{P}, I)$ is always a forest, but not a spanning forest because of the fake edges.
Intuitively, we will later replace the fake edges within each cluster $C$ with a spanning tree $T_C$ constructed recursively on $C$ and $\pi_C$; this will ensure that together with some inter-cluster edges, we have a spanning tree 
for $\widehat{G}$ and $\pi$. For each cluster in a HP, we consider it as an induced subgraph of $G$ and use the same notation for the induced subgraph and the cluster.  

\begin{lemma}[Preservable set lemma]
\label{lm:1+e_distort_preserve_set}
    Let $\mathfrak{H}$ be a pair-preserving $(\rho, \mu)$-HPF 
    of $G$ with $\rho = \tilde{O}(\ddim^3)$, $\mu = \tilde{O}(\ddim^3)$, $\e = \mu^{-O(\ddim)}$ and maximum degree $\smash{2^{\tilde{O}(\ddim)}}$.  
    Let $\mathbb{P}$ be an HP of the $(\rho, \mu)$-HPF $\mathfrak{H}$.
    Given 
    \begin{enumerate}
    \item any node $\widehat{G}$ at level $i$ in $\mathbb{P}$, 
    \item clustering $\mathcal{C}$ of $\widehat{G}$ at level $i - \smash{\log_{\mu}(1/\e)}$, 
    \item a path $\pi$ that is a shortest path of $\widehat{G} \cup \pi$, and 
    \item an $\e$-subcluster pair $(C_1, C_2)$ in $\widehat{G}$, 
    \end{enumerate}
    one can construct a preservable set $\cP$ with respect to $(\widehat{G}, \cC, \pi)$
    and a set of cycle-free inter-cluster edges $I$, so that in the sketch graph $H := \mathrm{sketch}(\cP, I)$:
    \begin{enumerate}
        \item \label{it:tree} 
        $H$ is a tree;
        \item \label{it:preserve_pair}
        $d_H(x, y) \leq d_{\widehat{G}}(x, y) + 44\cdot\e\mu^{i}$ for any pair of vertices $x \in C_1$ and $y \in C_2$;
        \item \label{it:bdd_diam} $d_H(u, v) \leq 10\mu^i$ for every pair of vertices $u$ and $v$ in $\widehat{G}$. 
    \end{enumerate}
\end{lemma}

\subsection{Constructing the spanning tree cover}
\label{SS:tree-cover}

To construct a $(1 + \e)$-stretch spanning tree cover, we use the procedure \textsc{SpanTreeCover}$(G)$. First, we use \Cref{lm:pair_preserve_HPF} to find a bounded-degree pair-preserving $(\rho, \mu)$-HPF $\mathfrak{H}$ of $G$ with $\rho = \tilde{O}(\ddim^3)$ and $\mu = \tilde{O}(\ddim^3)$.
Recall that for each cluster $C$ in a $\mu$-hierarchical partition $\mathbb{P}$, its \emph{$\e$-subclusters} are the clusters that are $\ell$ levels below $C$, where $\ell \coloneq \Ceil{\log_\mu (1/\e)}$ is the \emph{offset}.
For each hierarchical partition $\mathbb{P} = (H_0, H_1, \ldots H_{i_{\max}}) \in \mathfrak{H}$, we build $\ell$ spanning trees 
by recursively running \textsc{PathPreservingTree}$(\mathbb{P}, G_{i}, \varnothing)$ with $i \in [i_{\max} - \ell + 1, i_{\max}]$ and $G_i$ being the copy of~$G$ corresponding to the trivial partition at level $i$. 
Recall that $H_i$ is the trivial partition $\{V(G)\}$ for all $i \in [i_{\max} - \ell + 1, i_{\max}]$. 
For any cluster $\widehat{G}$ with the associated subcluster pair $(C_1,C_2)$ and any shortest path $\pi$ touching $\widehat{G}$, \textsc{PathPreservingTree}$(\mathbb{P}, \widehat{G}, \pi)$ finds a preservable set \EMPH{$\cP$} with respect to $(\widehat{G}, \cC, \pi)$ (recall that $\cC$ is a clustering of $\widehat{G}$ at $\ell$ levels below and $\cC$ contains $C_1, C_2$)
and a set of cycle-free inter-cluster edges \EMPH{$I$} satisfying the conditions in \Cref{lm:1+e_distort_preserve_set}. 
For each cluster $C \in \cC$, we recursively build a spanning tree of $G[C]$ by running \textsc{PathPreservingTree}$(\mathbb{P}, G[C], \pi_C)$ with $\pi_C$ being the path in $\mathcal{P}$ touching $C$. 
Then, the spanning tree of $\widehat{G}$ is the union of all output spanning trees of clusters in $\cC$ plus $I$.  
We output the set of all trees corresponding to each hierarchy. 

\begin{figure}[h!]
\begin{tcolorbox}
\internallinenumbers
    \paragraph{\textsc{PathPreservingTree}$(\mathbb{P}, \widehat{G}, \pi)$:~} 
    This procedure takes three arguments: (1) $\mu$-hierarchical partition~$\mathbb{P}$, (2) cluster $\widehat{G}$ in the hierarchy, and (3) path $\pi$. If $\pi$ is empty, we choose $\pi$ to be an arbitrary vertex in $\widehat{G}$. 
    It returns a spanning tree of $\widehat{G} \cup \pi$ following the steps below:
    \medskip
    \begin{enumerate}
        \item  
        If $\widehat{G}$ contains only one vertex, return $\pi$.
        \item  
        Let $(C_1, C_2)$ be the subcluster pair associated with the cluster $\widehat{G}$. 
        \item  
        Let $\mathcal{C}$ be the clustering of $\widehat{G}$ at $\log_{\mu}(1/\e)$ level below in the hierarchy $\mathbb{P}$.
        \item \label{line:find_preserve_set}  
        Construct preservable set $\mathcal{P}$ with respect to $(\widehat{G}, \mathcal{C}, \pi)$ and a set of cycle-free inter-cluster edges $I$ using \Cref{lm:1+e_distort_preserve_set}, given the pair $(C_1, C_2)$. 
        \item \label{it:recursive} 
        For each cluster $C \in \mathcal{C}$, let $T_C \gets$ \textsc{PathPreservingTree}$(\mathbb{P}, G[C], \pi_C)$. Here, $\pi_C$ is the (unique) path in $\mathcal{P}$ touching $C$.
        \item  \label{it:union} Return $T \coloneqq I \cup \bigcup_{C \in \mathcal{C}} T_C$.
    \end{enumerate}
    
    \paragraph{\textsc{SpanTreeCover}$(G)$:~} 
    This procedure takes a graph $G$ with doubling dimension $\ddim$, and returns a tree cover $\cT$ of $G$:
    \begin{enumerate}
        \item Construct a strong-diameter pair-preserving $(\rho, \mu)$-HPF $\mathfrak{H}$ using \Cref{lm:pair_preserve_HPF}.
        Here, $\rho = \tilde{O}(\ddim^3)$ and $\mu = \tilde{O}(\ddim^3)$.
        \item  For each hierarchy $\mathbb{P} = (H_1, H_2, \ldots H_{i_{\max}}) \in \mathfrak{H}$, construct $\ell \coloneqq \log_\mu(1/\e)$ spanning trees $T^i_{\mathbb{P}} \gets$ \textsc{PathPreservingTree}$(\cH, G_i, \varnothing)$ for $G_i$, for each $i \in [i_{\max} - \ell + 1 : i_{\max}]$,
        where $G_i$ is the cluster containing the whole graph in $H_{i}$.  Let $\mathcal{T}_{\mathbb{P}} \gets \set{T^i_\mathbb{P}}$. 
       \item  Return $\mathcal{T} \coloneqq \bigcup_{\mathbb{P} \in \mathfrak{H}}\mathcal{T}_{\mathbb{P}}$.
    \end{enumerate}

\end{tcolorbox}
\end{figure}

Given \Cref{lm:1+e_distort_preserve_set}, we show that \textsc{PathPreservingTree} returns a spanning tree $T$ of $\widehat{G} \cup \pi$ such that (1) the distance between the $\e$-subcluster pair in $T$ approximates their distance in $\widehat{G}$; and (2) the distance between any two vertices $u$ and $v$ in $T$ is upper bounded by the diameter of cluster $\widehat{G}$ (intuitively, tree $T$ is not too far away from a shortest path tree).

\begin{lemma}\label{lm:1+e-spanning-pathpreser} 
Suppose that $\mathbb{P}$ is an HP from a pair-preserving $(\rho, \mu)$-HPF $\mathfrak{H}$. 
Let $T$ be the tree returned by \textsc{PathPreservingTree}$(\mathbb{P}, \widehat{G}, \pi)$ for some cluster $\widehat{G}$ at level $i$, and let $(C_1,C_2)$ be the subcluster pair associated with $\widehat{G}$. 
Then:
\begin{enumerate}
    \item \label{it:spanning} 
    $T$ is a spanning tree of $\widehat{G} \cup \pi$.
    \item \label{it:1+e_distort_pair} 
    Let $(x, y)$ be any pair of vertices such that $x \in C_1$ and $y \in C_2$. Then, $d_T(x, y) \leq d_{\widehat{G}}(x, y) + 44\cdot \e\mu^i$. 
    \item \label{it:constant_distortion}  
    For every two vertices $u$ and $v$ in $\widehat{G}$, $d_T(u,v) \leq 10\mu^i$.
\end{enumerate}
\end{lemma}

\begin{proof}
We use induction to prove that $T$ is a tree. Specifically, we inductively show that $T$ is a spanning tree of $\widehat{G} \cup \pi$. If $\widehat{G}$ contains only one vertex, $T$ is the path $\pi$, and hence is a spanning tree of the single vertex in $\widehat{G}$ and $\pi$ (since $\pi$ is touching that vertex by definition). 
Let $\cC$ be the clustering of $\widehat{G}$ in $\log_\mu(1/\e)$ levels down the hierarchy $\mathbb{P}$. 
By induction, \EMPH{$T_C$} constructed in step \ref{it:recursive} of \textsc{PathPreservingTree} is a spanning tree of $\widehat{G}[C] \cup P_C$ 
that contains $P_C$. Let \EMPH{$\widehat{T}_C$} be the tree that contains $P_C$ and (fake) edges from each vertex $v \in C\setminus P_C$ to its nearest vertex in $P_C \cap C$, each edge has weight $10\e\mu^i$. Note that \EMPH{$\widehat{T}_C$} is a subgraph of $\Sketch(\mathcal{P}, I)$. 
Recall that the diameter of each cluster in $\cC$ is bounded by $\e\mu^i$.
Define 
\begin{equation*}
K \coloneqq \cup_{C\in \cC} T_C \quad \text{and} \quad
\widehat{K} \coloneqq \cup_{C\in \cC} \widehat{T}_C.
\end{equation*}
Let $I$ be the cycle-free inter-cluster edges from step~\ref{line:find_preserve_set}.
We observe that (a) $K\cup I$ is a tree if and only if $\widehat{K}\cup I$ is a tree,
and (b) $\widehat{K} \cup I$ is exactly $\Sketch(\mathcal{P}, I)$. 
The two observations together with \cref{it:tree} of \Cref{lm:1+e_distort_preserve_set} imply that $K\cup I$, which is $T$, is a tree.

To show that $T$ is a spanning tree of $\widehat{G} \cup \pi$, we use induction on the level of $\widehat{G}$ in $\mathbb{P}$. 
If $\widehat{G}$ contains only one vertex, $T = \pi$ is a spanning tree of $\widehat{G} \cup \pi$ (since $\pi$ touched $\widehat{G}$). 
Assume inductively that each $T_C$ in step $5$ of \textsc{PathPreservingTree} is a spanning tree of $G[C] \cup \pi_C$. 
Since all edges in $I$ are in $\widehat{G}$, 
and the definition of preservable set implies that every cluster in $\cC$ is touched by one path in $\cP$,
$T \coloneqq I \cup \bigcup_{C \in \cC} T_C$ is a spanning subgraph of $\widehat{G} \cup \pi$. 
Combine with the fact that $T$ is a tree, we conclude $T$ is a spanning tree of $\widehat{G} \cup I$.

Next, we prove \cref{it:1+e_distort_pair} and \cref{it:constant_distortion}, again by induction.  
Note that $H\coloneqq \widehat{K} \cup I$ is exactly the sketch graph provided by \Cref{lm:1+e_distort_preserve_set}. 
Fix a pair of vertices $u$ and $v$ in $\widehat{G}$.
Let \EMPH{$\pi_H(u, v)$} be the (unique) path from $u$ to $v$ in $H$.  
Let \EMPH{$\tilde{\pi}_T(u,v)$} be a (possibly non-simple) path from $u$ to $v$ in $T$ obtained as follows: 
for each fake edge $e' = (u',v')$ on $\pi_H(u,v)$ connecting a vertex $v'$ in a cluster $C$ to its closest vertex $u' \in \pi_C$, we replace $e'$ by the path from $v'$ to $u'$ in $T_C$.
Each fake edge $e'$ is the shortest path between $u'$ and $v'$ in $H$ because $H$ is a tree by \Cref{lm:1+e_distort_preserve_set}(\ref{it:tree}), and has length $10 \e \mu^i$ by definition. 
Thus by induction (using \cref{it:constant_distortion}), 
we have $d_{T_C}(u', v') \leq 10\e\mu^i = d_H(u', v')$, which implies that for any $u$ and $v$ in $\widehat{G}$,
\begin{equation} \label{eq:dtub}
d_T(u, v) ~\le~ w(\tilde{\pi}_T(u, v)) ~\leq~ w(\pi_H(u, v)) ~=~ d_H(u, v).
\end{equation}
For \cref{it:1+e_distort_pair}, 
by  \Cref{lm:1+e_distort_preserve_set}(\ref{it:preserve_pair}) and \Cref{eq:dtub}, we have $d_T(x, y) \leq d_H(x, y) \le d_{\widehat{G}}(x, y) + 44\e\mu^i$ for any pair $(x,y)$ from the subcluster pair. 
As for \cref{it:constant_distortion}, by \Cref{lm:1+e_distort_preserve_set}(\ref{it:bdd_diam}), $d_T(u,v) \le d_H(u,v) \le 10\mu^i$, given that $\e = \mu^{-O(\ddim)}$. 
\end{proof}

We next show that the \textsc{SpanTreeCover} algorithm returns a $(1 + \e\mu)$-stretch spanning tree cover of size $2^{O(\ddim)} \cdot \e^{-O(\ddim)}$. 
First we need an observation and a lemma.
We write \textsc{Tree} as a shorthand of \textsc{PathPreservingTree}. 

\begin{observation}
\label{obs:subtree}
For any cluster $C$ at level $i$ of graph $G_j$ such that $i \equiv j \pmod \ell$, \textsc{Tree}$(\mathbb{P}, G_j, \varnothing)$ recursively makes a call to \textsc{Tree}$(\mathbb{P}, G[C], \pi_C)$ for some $\pi_C$. 
Furthermore, let $T$ be the output of \textsc{Tree}$(\mathbb{P}, G_j, \varnothing)$, then the output of \textsc{Tree}$(\mathbb{P}, G[C], \pi_C$) is $T[C] \cup \pi_C$. 
\end{observation}

\begin{proof}
We use induction on the level $i$. 
For the basis of the induction when $i = j$, \Cref{obs:subtree} trivially holds by \Cref{lm:1+e-spanning-pathpreser}(\ref{it:spanning}). 
For the inductive case,
let $\widehat{C}$ be the cluster containing $C$ at level $i + \ell$. 
By the induction hypothesis, \textsc{Tree}$(\mathbb{P}, G_j, \varnothing)$ recursively makes call to \textsc{Tree}$(\mathbb{P}, G[\widehat{C}], \pi)$ for some $\pi$. 
From Step~\ref{it:recursive} of \textsc{Tree}, \textsc{Tree}$(\mathbb{P}, G[\widehat{C}], \pi)$ makes a recursive call to \textsc{Tree}$(\mathbb{P}, G[C], \pi_C)$ for some path $\pi_C$.

Furthermore, let $T_C$ be the output of \textsc{Tree}$(\mathbb{P}, G[C], \pi_C)$ and $T_{\widehat{C}}$ be the output of \textsc{Tree}$(\mathbb{P}, G[\widehat{C}], \pi)$ for some $\pi$. 
By the induction hypothesis, $T_{\widehat{C}}$ is $T[\widehat{C}] \cup \pi$. 
By \Cref{lm:1+e-spanning-pathpreser}(\ref{it:spanning}), $T_C$ is a spanning tree of $G[C] \cup \pi_C$. (Note that $T_C$ contains $\pi_C$.) 
On the other hand, by Step~\ref{it:union} of \textsc{Tree}, $T_C$ is a subtree of $T_{\widehat{C}}$, and hence is a subtree of $T$. 
Thus, $T_C = T[C] \cup \pi_C$. 
\end{proof}

\begin{lemma}
\label{lem:assigned-good-stretch}
Let $(u,v)$ be any vertex pair in $G$, and
let $\mathbb{P}$ be an HP preserving $(u,v)$ with {scale-$i$} cluster $C$ and subcluster pair $(C_1, C_2)$. 
Let $j \coloneqq i \bmod \ell$. 
Let $T$ be the tree $T_{\mathbb{P}}^j$ constructed by \textsc{SpanTreeCover}.
Then for any two vertices $x \in C_1$ and $y \in C_2$,
\[
\dist_G(x,y) \le \dist_{T}(x,y) \le (1+O(\rho\e)) \cdot \dist_G(x,y).
\]
\end{lemma}

\begin{proof}
From \Cref{obs:subtree}, \textsc{Tree}$(\mathbb{P}, G_j, \varnothing)$ recursively makes a call to \textsc{Tree}$(\mathbb{P}, G[C], \pi_C)$ for some $\pi_C$, which output $T[C] \cup \pi_C$. 
Since $T[C] \cup \pi_C$ is a subtree of $T$, 
by \Cref{lm:1+e-spanning-pathpreser}(\ref{it:1+e_distort_pair}), for any pair of vertices $x \in C_1$ and $y \in C_2$,
\[
d_T(x,y) ~=~ d_{T[C] \cup \pi}(x, y) ~\leq~ d_{G[C]}(x, y) + 44\cdot \e\mu^i.
\]
Because $\mathbb{P}$ preserves $(u,v)$ and $(x,y) \in C_1 \times C_2$,
$d_G(x, y) = d_{G[C]}(x, y) > \mu^i/\rho$.
Therefore
\begin{align*}
d_T(x,y) ~\leq~ d_{G[C]}(x, y) + 44\cdot \e\mu^i ~\leq~ \Paren{1+ 44\cdot \rho\e} \cdot d_{G}(x, y). 
\end{align*}
\aftermath
\end{proof}

\begin{theorem}
\textsc{SpanTreeCover}$(G)$ produces a $(1 + \e)$-stretch tree cover $\mathcal{T}$ of $G$ with size $ \e^{-O(\ddim)}$. 
\end{theorem}

\begin{proof}
By \Cref{lm:1+e-spanning-pathpreser}, $\mathcal{T}$ is a set of spanning trees of $G$. 
The size of $\mathcal{T}$ is equal to the size of $\mathfrak{H}$ times $\ell \coloneqq \Ceil{\log_{\mu}(1/\e)}$ since each HP produces $\ell$ trees. 
The size of $\mathfrak{H}$ is $\e^{-O(\ddim)} \cdot 2^{\tilde{O}(\ddim)}$. 
Thus, $|\mathcal{T}| = \ell \cdot |\mathfrak{H}| = \e^{-O(\ddim)}$ assuming $\e^{-1} \ge \ddim$.
We now show that for each pair  of vertices $(u,v)$ in $G$, there exists a tree $T$ in $\mathcal{T}$, such that  $d_T(u,v) \le (1 + O(\rho\e)) \cdot d_G(u, v)$; scaling $\e$ then proves the theorem.

By \Cref{def:pair_HPF} and \Cref{lm:pair_preserve_HPF}, there is an HP $\mathbb{P}$ in $\mathfrak{H}$ preserving $(u,v)$ with a cluster $C$ at level $i$ such that the associated ($\rho$-well-separated) subcluster pair $(C_u,C_v)$ contains $u$ and $v$ respectively, and 
$d_G(u, v) = d_{G[C]}(u, v) > \mu^i/\rho$.
Let $\mathbb{P} = (P_1,  \ldots P_{i_{\max}})$ and $j$ be the integer in $(i_{\max} - \ell, i_{\max}]$ such that $j \equiv i \pmod \ell$. 
Let $T$ be the output tree of \textsc{Tree}$(\mathbb{P}, P_j, \varnothing)$. 
By \Cref{lem:assigned-good-stretch}, 
for any two vertices $x \in C_u$ and $y \in C_v$, in particular $u \in C_u$ and $v \in C_v$,
\[
\dist_G(u,v) \le \dist_{T}(u,v) \le (1+O(\rho\e)) \cdot \dist_G(u,v).
\]
This proves the theorem. \qed

\end{proof}
Here, we note the conditions we need for $\e$. In \Cref{lm:pair_preserve_HPF}, we require $\e^{-1} \geq \Tilde{\Omega}(\ddim^{3})$ and in \Cref{lm:1+e_distort_preserve_set}, we have $\e = \mu^{-O(\ddim)}$. Combining these conditions, we find that $\e = 2^{-\Tilde{O}(\ddim)}$ since $\mu = \Tilde{O}(\ddim^3)$. 

For \Cref{thm:main} (as well as \Cref{thm:spanning-main}, \Cref{thm:routing} and \Cref{thm:oracle}) to hold, we scale down $\e$, which results in a additive constant loss in the exponent of $\e$. Specifically, we set $\delta = \e \cdot 2^{-\Tilde{O}(\ddim)}$, and apply algorithm \textsc{SpanTreeCover} for stretch $(1 + \delta)$ to achieve the desired result.
\section{Preservable Set}
\label{S:preservable-construct}

\subsection{Preservable set construction}
\label{sssec:preservable_set}

We now describe how to construct a preservable set $\cP$ with respect to $(\widehat{G}, \mathcal{C}, \pi)$ and a corresponding set of inter-cluster edges $I$ (which of each has its two endpoints in different clusters) such that given an $\e$-subcluster pair $(C_1, C_2)$ in $\widehat{G}$, the distance between any $x\in C_1$ and $y \in C_2$ in 
$\mathrm{sketch}(\cP, I)$ is a $(1 + O_{\ddim}(\e))$-approximation of their distance in $\widehat{G}$. 
In the next subsection we prove the correctness of the construction by establishing the three properties of $H$ in the preservable set lemma (\Cref{lm:1+e_distort_preserve_set}).

\paragraph{Algorithm.}
For any path $\pi'$, denote by \EMPH{$\cC[\pi']$} the set of clusters touched by $\pi'$ in $\cC$. 
For any set of paths~$\cS$,  denote by $\cC[\cS]$ the set of clusters touched by any path in $\cS$. 
We construct a preservable set $\cP$ and a set $I$ of inter-cluster edges that prove \Cref{lm:1+e_distort_preserve_set} as follows.

\paragraph{Step~1.} 
Initialize $\cP \gets \{\pi\}$. Let $\EMPH{$P_{xy}$} = (x_1, x_2, \ldots x_k)$ be a shortest path from $x$ to $y$ in $\widehat{G}$. 

\begin{itemize}
\item
If $\cC[P_{xy}] \cap \cC[\pi] = \varnothing$ (see \Cref{fig:non_intersect}), let $P' = (x'_1, x'_2, \ldots)$ be a shortest path from $x$ to $\pi$ in $G$. 
Let \EMPH{$x'_{j_2}$} 
be the first vertex in $P'$ that is in (some cluster in) $\cC[\pi]$ and \EMPH{$x'_{j_1}$} be the last vertex in $P'[x'_1:x'_{j_2}]$ that is in (some cluster in) $\cC[P_{xy}]$. 
We add the path $P'' = (x'_{j_1 + 1}, x'_{j_1 + 2}, \ldots x'_{j_2 - 1})$ to $\cP$ and then add $(x'_{j_2 - 1}, x'_{j_2})$ to $I$. 
Then we add $P_{xy}$ to $\cP$ and $(x'_{j_1}, x'_{j_1 + 1})$ to $I$.

\begin{figure}[t]
    \centering
    \includegraphics[width=8cm] {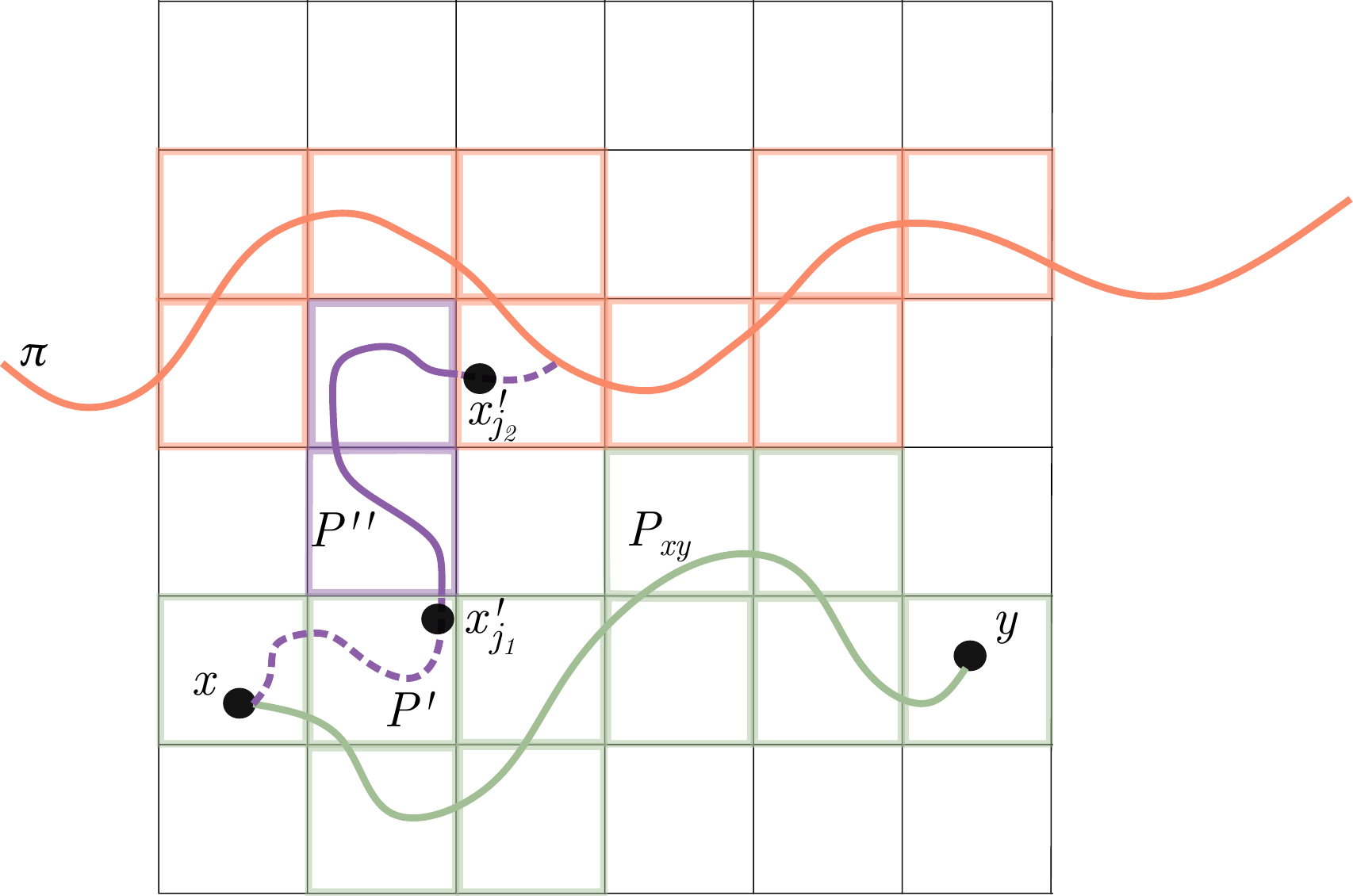}
    \caption{$\cC[\pi]$ and $\cC[P_{xy}]$ do not intersect. In this case, we add $P'$ (the solid purple path) to $\cP$. Note that $P''$ is a subpath of the shortest path $P'$ from $x$ to $\pi$.}
    \label{fig:non_intersect}%
\end{figure}

\item
Otherwise $\cC[P_{xy}] \cap \cC[\pi] \neq \varnothing$ (see \Cref{fig:intersect}), let $x_{j_3}$ be the first vertex in $P_{xy}$ that is in $\cC[\pi]$ and let $x_{j_4}$ be the last vertex in $P_{xy}$ that is in some cluster touched by either $\pi$ (see \Cref{fig:intersect} (a)) or $P_{xy}[x_1:x_{j_3 - 1}]$ (see \Cref{fig:intersect} (b)). Add $(x_1, x_2, \ldots, x_{j_3 - 1})$ and $(x_{j_4 + 1}, x_{j_4 + 2}, \ldots, x_k)$ to $\cP$, then add $(x_{j_3 - 1}, x_{j_3})$ and $(x_{j_4}, x_{j_4 + 1})$ to $I$.
\end{itemize}

\begin{figure}[h!]
\small
\centering
\includegraphics[width=8cm] {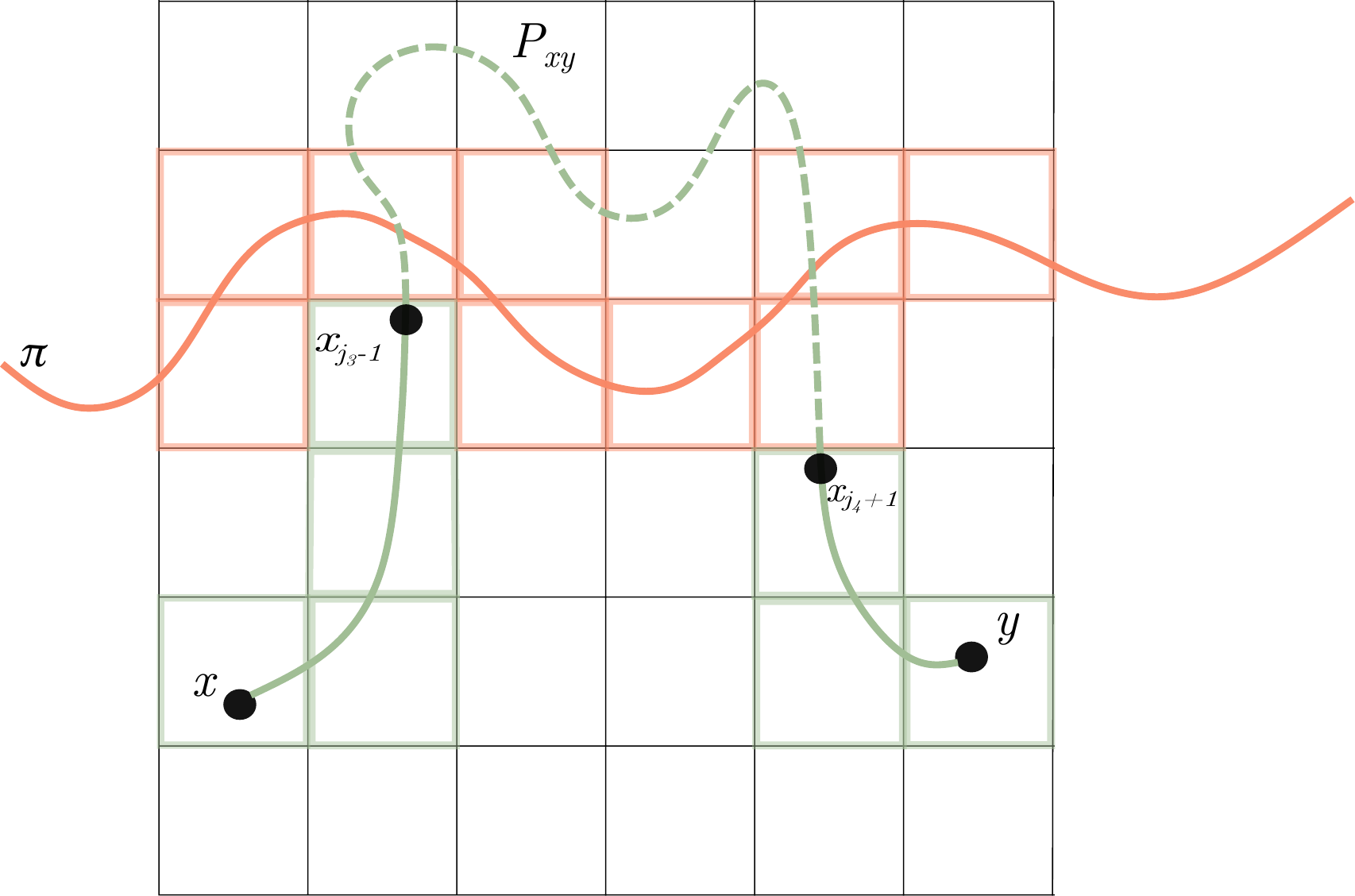}
\includegraphics[width=8cm] {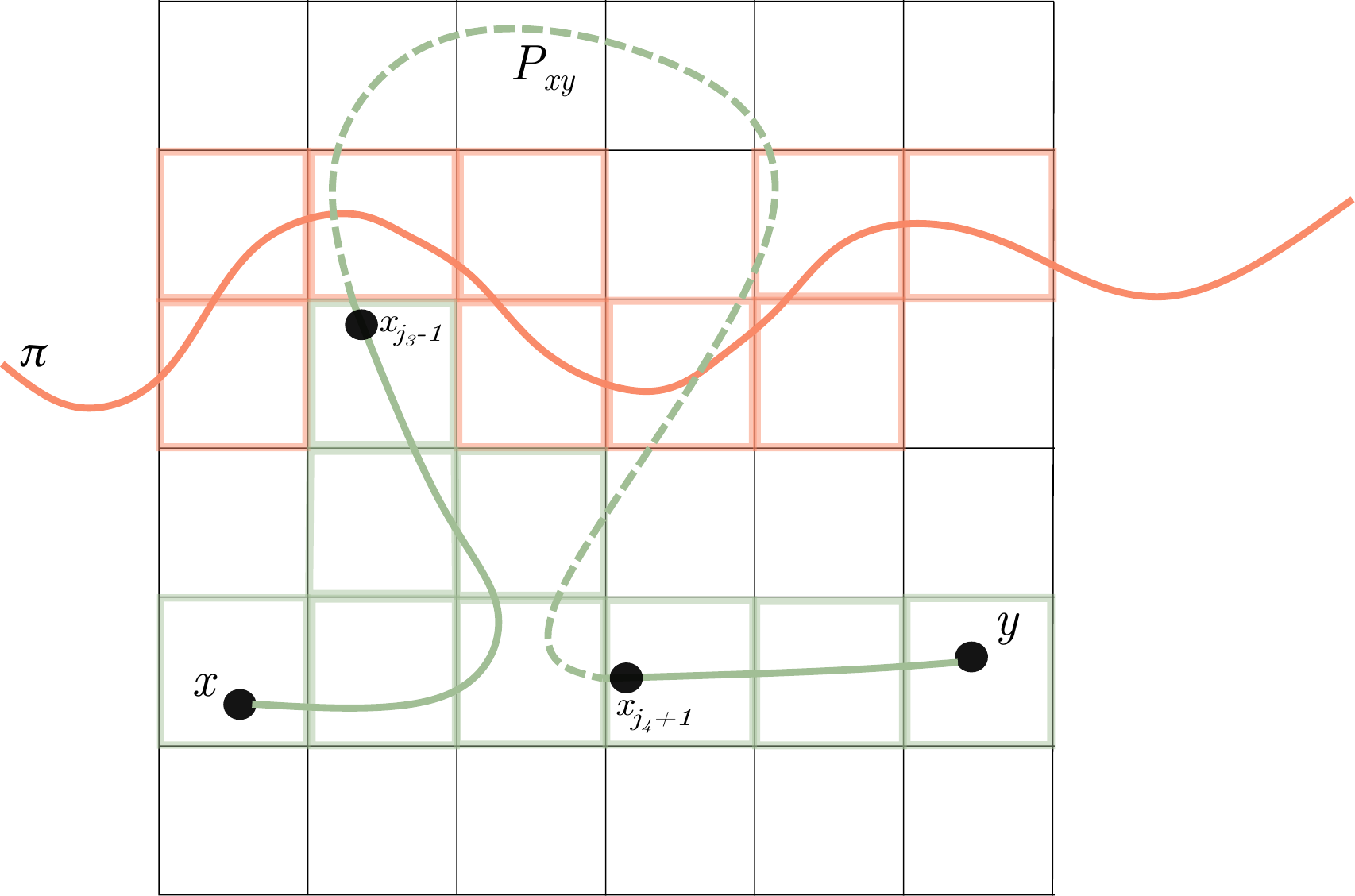}
\caption{$\cC[\pi]$ and $\cC[P_{xy}]$ intersect. In this case, we add $P_{xy}[x_1:x_{j_3 - 1}]$ and $P_{xy}[x_{j_4 + 1:x_k}]$ to $\cP$. 
The left figure indicates the case that $x_{j_4}$ (the vertex before $x_{j_4 + 1}$ in $P_{xy}$) is in some cluster in $\cC[\pi]$.  
The right figure indicates the case that $x_{j_4}$ is in some cluster in $\cC[P_{xy}[x_1:x_{j_3 - 1}]]$.}
\label{fig:intersect}
\label{fig:detour-code}
\end{figure}

\paragraph{Step~2.} 
Given a set of paths $\cP$ and inter-cluster edges $I$ from the last step, we will add include more paths via the following \EMPH{gluing procedure}: to add a path $Q = (v_1, v_2, \ldots )$ to $\cP$, find the first vertex $v_j$ in $Q$ that is in $\cC[\cP]$; then, add $(v_1, v_2, \ldots v_{j - 1})$ to $\cP$ and add the edge $(v_{j - 1}, v_j)$ to $I$. 
(If there is no such vertex $v_j$, i.e., $Q$ does not touch any cluster in $\cC[\cP]$, we simply add $Q$ to $\cP$.)

The second step runs in $\log_\mu(1/\e) + 1$ iterations.
Recall that $\widehat{G}$ is at level $i$. 
Let $\mathbb{P}_j$ be the partition of $\widehat{G}$ at level $j \in [i - \ell + 1 : i]$. 
For $j$ from $i$ down to $i - \ell$, pick an arbitrary representative vertex $r(C)$ from each cluster $C$ in $\mathbb{P}_{j}$, let \EMPH{$C'$} be the parent cluster of $C$ in $\mathbb{P}_{j+1}$ and $r'$ be the corresponding representative vertex. 
We find a shortest path from $r_C$ to $r'$ in $C'$. (For the first iteration when $j = i$), we simply find the shortest path from $r_C$ to $\pi$.) 
We denote this shortest path \EMPH{$P_C$}.  
Then, we include $P_C$ to $\cP$ via the gluing procedure.
This concludes the description of the algorithm for constructing $\cP$ and $I$.

\subsection{Preservable set analysis: Proof of \Cref{lm:1+e_distort_preserve_set}}

\paragraph{Item~1: \boldmath{$H$} is a tree.}
Consider the $\cP$ and $I$ constructed from the previous subsection, and let $H \coloneqq \mathrm{sketch}(\cP, I)$. 
We show that $H$ satisfies the items of \Cref{lm:1+e_distort_preserve_set}.
We first show that $H$ is indeed a cycle-free sketch graph (and in fact, a tree).

\begin{claim}
\label{clm:pi}
    $\cP$ is a preservable set and $I$ is a set of inter-cluster edges.  As a result, $H$ is well-defined.
\end{claim}

\begin{proof}
    We first show \Cref{it:cluster_disjoint} of \Cref{def:preservable}. 
    Note that, to show \Cref{it:cluster_disjoint}, it is sufficient to point out that each cluster is touched by at least one path in $\cP$ and there are no two paths in $\cP$ touching the same cluster (we refer to as \EMPH{cluster-disjoint}). By the last iteration of the second step, each cluster is touched by at least one path. Thus, we only need to show that all paths in $\cP$ are cluster-disjoint, as this condition also implies paths in $\cP$ are vertex disjoint. 
    In the first step, if $\cC[P_{xy}] \cap \cC[\pi] = \varnothing$ are disjoint, we add both $P_{xy}$ and $\pi$ to $\cP$. Then, we add that path $P''$ to $\cP$. Observe that the path $P''$ does not touch any cluster in $\cC[P_{xy}]$ and $\cC[\pi]$. 
    Then, $\cC[P''], \cC[\pi]$ and $\cC[P_{xy}]$ are disjoint. Similarly, when $\cC[P_{xy}] \cap \cC[\pi] \neq \varnothing$, the two paths added to $\cP$ are cluster-disjoint. 
    In the second step, every path added by the gluing procedure is cluster disjoint with all the paths in $\cP$. Using induction, we conclude the output $\cP$ is cluster-disjoint.

    Observe that $\cP$ contains $\pi$, and hence satisfies \Cref{it:pi} of \Cref{def:preservable}. 
    In any step, the path $P$ touching cluster $C$ is always a shortest path in a supergraph containing $P$ and $G[C]$. Thus, the \Cref{it:shortest_path} also holds. 
    $I$ is a set of inter-cluster edges by construction.
\end{proof}

Now we show that \Cref{lm:1+e_distort_preserve_set}(\ref{it:tree}) holds.

\begin{claim}
    $H$ is a tree.
\end{claim}

\begin{proof}
Imagine the sketch graph $H$ is built simultaneously with $\cP$ and $I$ during the preservable set construction.
We argue that every time an edge is added into $I$, we connect two disjoint components in $H$ and thus $H$ remains a tree.

The graph $H$ starts with a single path $\pi$ in $\cP$ and fake edges from every vertex in every cluster touched by $\pi$, which is obviously cycle-free because the clusters are disjoint, and connected because every vertex can reach path $\pi$ by a fake edge.  Thus $H$ started as a tree.
In each step of the algorithm we add a path $Q$ where the clusters $\cC[Q]$ touched by $Q$ is disjoint from any clusters touched by any paths in the existing $\cP$, so that $\cP$ remains cluster-disjoint.
This means that in $H$ we create a new connected component in $H$ for $Q$, consisting of the path $Q$ itself and all the fake edges from every vertex in every cluster touched by $Q$.
Finally we add an edge into $I$, connecting $Q$ to some vertex in a cluster touched by the existing $\cP$.
Such an edge connects two components in $H$ into one, and thus keeping $H$ cycle-free.
Since $H$ was connected before adding the component associated with $Q$, the updated $H$ remains connected after the inter-cluster edge is inserted.
As a result, the final sketch graph $H$ must be cycle-free and connected, thus a tree.
\end{proof}

\paragraph{Item~2: \boldmath{$H$} preserves subcluster pair distance.}

Before proving  \Cref{lm:1+e_distort_preserve_set}(\ref{it:preserve_pair}), we show the following useful observation.
\begin{observation}
    \label{obs:same_clust_dist}
    For any two vertices $u$ and $v$ in a same cluster $C \in \cC$, $d_H(u, v) \leq 21\e\mu^i$.   
\end{observation}

\begin{proof}
    Recall that $\pi_C$ is the path in $\cP$ touching $C$. Define $u'$ be the closest vertex to $u$ in $\pi_C \cap C$; define $v'$ similarly. 
    By the construction of $\mathrm{sketch}(\cP, I)$, $d_H(u, u') = d_H(v, v') = 10\e\mu^i$. Since $\pi_C$ is a shortest path in $G[C] \cap \pi_C$ and both $u'$ and $v'$ are in $C$, $d_{\pi_C}(u', v') \leq \mathrm{diam}(G[C]) \leq \e\mu^i$. 
    This shows $d_H(u', v') \leq d_{\pi_C}(u', v') \leq \e\mu^i$ as $\pi_C$ is a subgraph of $H$. Using the triangle inequality, we obtain:
    \begin{equation*}
        d_H(u, v) \leq d_H(u, u') + d_H(u', v') + d_H(v', v) \leq 21\e\mu^i.
    \end{equation*}
    \aftermath
\end{proof}

\noindent We now establish \Cref{it:preserve_pair}.

\begin{claim}
    $d_H(x, y) \leq d_{\widehat{G}}(x, y) + 44\e\mu^i$ for any pair of vertices $x \in C_1$ and $y \in C_2$.
\end{claim}

\begin{proof}
If $\cC[P_{x, y}] \cap \cC[\pi] = \varnothing$,
then we add a shortest path from $x$ to $y$ in $\widehat{G}$ to $\cP$, hence $d_H(x, y) = d_{\widehat{G}}(x, y)$. We henceforth assume that $\cC[P_{x, y}] \cap \cC[\pi] \neq \varnothing$, and we consider the following two cases. 
        
\noindent\\{\em Case 1:  $x_{j_4}$ is in some cluster in $\cC[\pi]$.~} 
Let $z_3$, $z_4$ be two arbitrary vertices in $\pi$ that are in the same cluster as $x_{j_3}$ and $x_{j_4}$, respectively. 
Then, both $d_H(x_{j_3}, z_3)$ and  $d_H(x_{j_4}, z_4)$ are at most $21\e\mu^i$ by \Cref{obs:same_clust_dist}. Since $x_{j_3}$ and $z_3$ are in the same cluster, $d_{\widehat{G}}(x_{j_3}, z_3) \leq \e\mu^i$; similarly, $d_{\widehat{G}}(x_{j_4}, z_4) \leq \e\mu^i$. 
By the triangle inequality, we have: 
\begin{equation*}
\begin{split}
    d_H(x, y) &\leq d_H(x, x_{j_3}) + d_H(x_{j_3}, x_{j_4}) + d_H(x_{j_4}, y)\\
    &\leq d_{\widehat{G}}(x, x_{j_3}) + d_H(x_{j_3}, z_3) + d_H(z_3, z_4) + d_H(x_{j_4}, z_4) + d_{\widehat{G}}(x_{j_4}, y)\\
    &\leq d_{\widehat{G}}(x, x_{j_3}) + d_{\widehat{G}}(x_{j_4}, y) + d_{\widehat{G}}(z_3, z_4) + 42\e\mu^i \qquad\text{(since $x'_i, x'_j \in \pi$ and $\pi \in \cP$)}\\
    &\leq d_{\widehat{G}}(x, x_{j_3}) + d_{\widehat{G}}(x_{j_4}, y) + d_{\widehat{G}}(z_3, x_{j_3}) + d_{\widehat{G}}(x_{j_3}, x_{j_4}) + d_{\widehat{G}}(x_{j_4}, z_4) + 42 \e\mu^i\\
    &\leq d_{\widehat{G}}(x, y) + 44\e\mu^i
\end{split}
\end{equation*}
where the last inequality holds since $x_{j_3}$ and $x_{j_4}$ are on the shortest path $P_{x,y}$ from $x$ to $y$ in $\widehat{G}$.

\noindent\\{\em Case 2: $x_{j_4}$ is in some cluster in $\cC[P[x_1:x_{j_3}]]$.~}
Let $x_g$ be a vertex in $P[x_1:x_{j_3}]$ that is in the same cluster as $x_{j_4}$. 
Observe that $d_{\widehat{G}}(x, x_g) + d_{\widehat{G}}(x_g, x_{j_4}) + d_{\widehat{G}}(x_{j_4}, y) = d_{\widehat{G}}(x, y)$. Since $P[x_1:x_g]$ and $P[x_{j_4}:y]$ are in $H$ (because $P[x_{j_4 + 1}:y]$ is in $\cP$ and $(x_{j_4}, x_{j_4 + 1}) \in I$), $d_H(x, x_g) \leq d_{\widehat{G}}(x, x_g)$ and $d_H(x_{j_4}, y) \leq d_{\widehat{G}}(x_{j_4}, y)$.
By the triangle inequality, we get:
\begin{equation*}
\begin{split}
    d_H(x, y) &\leq d_H(x, x_g) + { \underbrace{d_H(x_g, x_{j_4})}_{\leq 21\e\mu^i \text{~by \ref{obs:same_clust_dist}}} } + d_H(x_{j_4}, y) \\
    &\leq d_{\widehat{G}}(x, x_g) + 21\e\mu^i + d_{\widehat{G}}(x_{j_4}, y) \leq d_{\widehat{G}}(x, y) + 21\e\mu^i,
\end{split}
\end{equation*}
implying that $d_H(x, y) \leq d_{\widehat{G}}(x, y) + 44\e\mu^i$. 
\end{proof}

\paragraph{Item~3: \boldmath{$H$} has small diameter.}
Next we show that \Cref{lm:1+e_distort_preserve_set}(\ref{it:bdd_diam}) holds: for every pair of vertices $u$ and $v$ within the cluster, their distance in $H$ is upper bounded by constant times the cluster diameter.

\begin{claim}
    \label{clm:smll-diam}
    For any $u$ and $v$ in $\widehat{G}$, $d_H(u, v) \leq 10\mu^i$.
\end{claim}

\begin{proof}
We will show that for any vertex $u \in \widehat{G}$, 
\begin{equation}
\label{eq:dist_u_pi}
    d_H(u, \pi) ~\leq~ 4\mu^i + O_{\ddim}(\log_{\mu} (1/\e) \cdot \e\mu^i).
\end{equation}
Since $H$ is a spanning tree of $\widehat{G} \cup \pi$ and there is no edge in $I$ connecting $\widehat{G}$ to $\pi \setminus \widehat{G}$,
traversing along a shortest path $\pi$ within $\widehat{G}$ costs at most the diameter $\widehat{G}$, which is $\mu^i$.
Thus $d_H(u, v) \le d_H(u, \pi) + d_H(v, \pi) + \mu^i$, and the claim is proved by assuming $O_{\ddim}(\log_{\mu} (1/\e)) < 1/\e$, which is true if $\e = \mu^{-O(\ddim)}$. 

\smallskip
We go over all possible cases of $u$. 
By the definition of preservable set, each cluster in $\cC$ is touched by at least one path in $\cP$. Then, it is sufficient to show \Cref{eq:dist_u_pi} for every vertex in $\cC[\cP]$. 
If $u \in \cC[\pi]$, $d_H(u, \pi) \leq 10\e\mu^i$ by construction of $H$. We henceforth assume that $u \notin \cC[\pi]$. Let $C_u$ be the cluster containing $u$. 
Let $\cP_1$ be the set $\cP$ after the first step. 
We first consider the case when $u$ is in some cluster in $\cC[\cP_1]$ which are dealt with by the first step of the construction of $H$.  

\medskip
\noindent \newline{\em First case of Step 1: $\cC[P_{xy}] \cap \cC[\pi] = \varnothing$.~}  
In this case, recall that we add a subpath $P'' = (x'_{j_1 + 1}, \ldots x'_{j_2 - 1})$ of $P'$ to $\cP$ in which $P' = (x'_1, x'_2, \ldots)$ is a shortest path from $x$ to $\pi$ in $\widehat{G}$, $x'_{j_2}$ is the first vertex in $P'$ that is in $\cC[\pi]$ and $x'_{j_1}$ is the last vertex in $P'[x'_1:x'_{j_2}]$ that is in  $\cC[P_{xy}]$. If $u$ is in some cluster in $\cC[P'']$, let $u'$ be a vertex in $P''$ that is in $C_u$. By \Cref{obs:same_clust_dist}, $d_H(u, u') \leq 21\e\mu^i$. On the other hand, $d_H(x'_{j_2}, \pi) \leq 10\e\mu^i$ by the construction of the sketch graph $H$. 
Since $u'$ and $x'_{j_2}$ are in $P''$ which is a subgraph of $H$, $d_H(u', x'_{j_2}) \leq \mu^i$. 
Thus, 
\begin{equation}
\label{eq:non-empty-intersect}
\begin{split}
    d_H(u, \pi) 
    ~\leq~ d_H(u, u') + d_H(u', x'_{j_2}) + d_H(x'_{j_2}, \pi) 
    ~\leq~ 21\e\mu^i + \mu^i + 10\e\mu^i
    ~\leq~ \mu^i + 31\e\mu^i.
\end{split}
\end{equation}
If $u \in \cC[P_{xy}]$, let $u''$ be a vertex in $P_{xy}$ that is also in $C_u$ and $z$ be a vertex in $P_{xy}$ that is also in the same cluster with $x'_{j_1}$. 
Since $u''$ and $z$ are in $P_{xy}$ and $P_{xy}$ is a subgraph of $H$, $d_H(u'', z) \leq d_{P_{xy}}(u'', z) \leq d_{\widehat{G}}(x, y) \leq \mu^i$.  Hence, 
\begin{equation*}
\begin{split}
    d_H(u, \pi) &\leq { \underbrace{d_H(u, u'')}_{\leq 21\e\mu^i \text{~by \ref{obs:same_clust_dist}}} } + d_H(u'', z) + { \underbrace{d_H(z, x'_{j_1})}_{\leq 21\e\mu^i 
    \text{~by \ref{obs:same_clust_dist}}} } + d_H(x'_{j_1}, \pi)\\
    &\leq 21\e\mu^i + \mu^i + 21\e\mu^i + \mu^i \leq 2\mu^i + 42\e\mu^i.
\end{split}
\end{equation*}

\noindent \newline{\em Second case of Step 1: $\cC[P_{xy}] \cap \cC[\pi] \ne \varnothing$.~} 
 Recall that in this case, we add two subpaths $P_1' = (x_1, \ldots, x_{j_3 - 1})$ and $P_2' = (x_{j_4 + 1}, \ldots, x_k)$ to $\cP$, then add $(x_{j_3 - 1}, x_{j_3})$ and $(x_{j_4}, x_{j_4 + 1})$ to $I$ with $x_{j_3}$ being the first vertex in $P_{xy}$ that is in $\cC[\pi]$ and $x_{j_4}$ being the last vertex in $P$ that is in some cluster touched by either $\pi$ or $P[x_1:x_{j_3 - 1}]$.  
If $u \in \cC[P_1']$, let $u''$ be a vertex in $P_1'$ that is also in $C_u$. By \Cref{obs:same_clust_dist}, $d_H(u, u'') \leq 21\e\mu^i$. Since $u'', x_{j_3}$ are in $P_{xy}$, which is a shortest path in $\widehat{G}$, $d_{P_{xy}}(u'', x_{j_3}) \leq \mathrm{diam}(\widehat{G}) \leq \mu^i$. 
\begin{equation*}
    d_H(u, \pi) ~\leq~ d_H(u, u'') + d_H(u'', x_{j_3}) + d_H(x_{j_3}, \pi) ~\leq~ 21\e\mu^i + \mu^i + 10\e\mu^i ~\leq~ \mu^i + 31\e\mu^i.
\end{equation*}
For the case $u \in \cC[P_2']$, using similar argument, we obtain $d_H(u, \pi) \leq \mu^i + 31\e\mu^i$.

\smallskip
We have proven that for each $u$ in some cluster in $\cC[\cP_1]$,
\begin{equation}
\label{eq:step1}
    d_H(u, \pi) \leq 2\mu^i + 42\e\mu^i.
\end{equation}

\noindent \emph{The remaining cases (Step 2).} 
Recall that during Step~2, 
we process the clusters from level $i$ down to $i - \ell$, and each cluster $C$ in $\mathbb{P}_j$ has an arbitrary chosen representative $r_C$.
We inductively show that after the $j$-th iteration, for any $u$ in some cluster in $\cC[\cP]$, 
\[
d_H(u, \pi) \leq f(j) \quad \text{where} \quad \EMPH{$f(j)$} \coloneqq 2\mu^i + 52\e\mu^i + \sum_{k\in [j: i]} \mu^{k}  + (i - j)\cdot \mu^{O(\ddim)} \cdot \e\mu^i.
\]
We call this the \EMPH{outer induction hypothesis}. 
Notice that we prove \Cref{eq:dist_u_pi} after the $(i-\ell)$-th iteration, because for any $u \in \widehat{G}$, 
\begin{equation*}
\begin{split}
     d_H(u, \pi) 
     &\leq f(i - \ell) \coloneqq 2\mu^i + 52\e\mu^i + { \underbrace{\sum_{k\in [i - \ell: i]} \mu^{k}}_{\leq 2\mu^i \text{~by geometric sum}} }
     + \ell \cdot \mu^{O(\ddim)} \cdot \e\mu^i \\
     &\leq 4\mu^i + O_{\ddim}(\log_{\mu}(1/\e)  \cdot \e\mu^i).
\end{split}
\end{equation*}
We first show a simple but useful observation.

{
\begin{observation}
\label{obs:dist_cluster_path}
    Assume a path $Q$ from $r$ is added to $\cP$ by the gluing procedure in Step~$2$. Then, for any $u$ in any cluster in  $\cC[Q]$, $d_H(u, \pi) \leq d_H(r, \pi) + 10\e\mu^i$.
\end{observation}

\begin{proof}
Let $e$ be the inter-cluster edge in $I$ that is added together with $Q$. 
Since $\pi \in \cP$, $Q \circ e$ is the prefix of the unique path from $r$ to $\pi$ in $\mathrm{sketch}(\cP \cup \{Q\}, (I \cup \{e\})$ (and also in $H$). Thus, for any $z \in Q$, $d_H(z, \pi) \leq d_H(r, \pi)$.
Let $C_{u}$ be the cluster containing $u$ and $z$ be the closest vertex to $u$ in $Q \cap C_{u}$. 
By the construction of $H$, $d_H(u, z) = 10\e\mu^i$. 
Then,
\begin{equation*}
    d_H(u, \pi) ~=~ d_H(u, z) + d_H(z, \pi) 
    ~\leq~ d_H(r, \pi) + 10\e\mu^i. 
\end{equation*} 
\aftermath
\end{proof}
}

\smallskip
\noindent \emph{Base case:~} 
In the first iteration of Step~$2$, we only add a prefix $Q_C$ of the shortest path $P_C$ from the representative $r_C$ of $C \coloneqq \widehat{G}$ to $\pi$ by the gluing process. We have: 
\begin{equation*}
\begin{split}
    d_H(r_C, \pi) &\leq d_H(r_C, \cC[\cP_1]) + { \underbrace{\max_{z \in C, C \in \cC[\cP_1]}d_H(z, \pi)}_{\leq 2\mu^i + 42\e\mu^i \text{~by Eq. \ref{eq:step1}}} }\\ 
    &\leq w(Q_C) + 2\mu^i + 42\e\mu^i \leq 2\mu^i + 42\e\mu^i + \mu^i.
\end{split}
\end{equation*}
Using \Cref{obs:dist_cluster_path}, we obtain that for any $u$ in any cluster in $\cC[Q_C]$, $d_H(u, \pi) \leq d_H(r_C, \pi) + 10\e\mu^i \leq f(i)$. 
Thus, the base case holds.

\smallskip
\noindent \emph{Inductive case:~}
Let $C'$ be a cluster in $\mathbb{P}_{j + 1}$ ($j < i$) and $r'$ be the representative of $C'$. 
For each child cluster $C_s$ of $C'$ (where $s$ ranges in $[1:\mu^{O(\ddim)}]$ by the degree bound on each cluster),
let $r_s$ be the representative of $C_s$ and $Q_s$ be the shortest path from each $r_s$ to $r'$ in $C'$. 
Without loss of generality, assume that $Q_1, Q_2, \ldots$ are added to $\cP$ using the gluing procedure in this order, and let $Q'_s$ be the prefix of $Q_s$ that are added to $\cP$. 
We show again by induction that for any $s$, for any vertex $u_s$ in $Q'_s$, 
\[
d_H(u_s, \pi) \leq f(j + 1) + d_{C'}(u_s, r') + 11 \cdot s \cdot \e\mu^i.
\]
We call this the \EMPH{inner induction hypothesis}.
(In this proof, we simplify the notation by using $d_{C'}$ as the distance function in $G[C']$ for any cluster $C'$.)
Observe that before we glue $Q_1$ to $\cP$, only the paths in $\cP_1$ or the paths from representatives of ancestors of $C'$ (including $C'$ itself) can touch $C'$. 
Let $\cP_{j + 1}$ be the set $\cP$ after iteration $j+1$.

\begin{itemize}
\item
\emph{For the base case}:  Let $e = (z, z')$ be the edge added to $I$, where $z \in V(Q_1)$. 
Hence, $z'$ is in some cluster in $\cP_{j + 1}$.
By outer induction hypothesis, $d_{H}(z', \pi) \leq f(j + 1)$.  
For any vertex $u_1 \in Q'_1$, the shortest path from $u_1$ to $r'$ in $C'$ contains $z'$, so $d_{H}(u_1, z') = d_{C'}(u_1, z') \le d_{C'}(u_1, r')$.
Thus, we get $d_H(u_1, \pi) \leq d_H(u_1, z') + d_H(z', \pi) \leq d_{C'}(u_1, r') + f(j + 1)$. 

\item
\emph{For the inductive case:}
Consider $Q_s$ with $s > 1$. If $Q_s$ touches a cluster in $\cP_{j + 1}$ first, we use the same argument as $Q_1$. 
Otherwise, $Q_s$ touches a cluster in $\cC[Q'_{s'}]$ first for some $s' < s$. 
Let $(b_1, b_2)$ be the inter-cluster edge added to $I$ after we add $Q'_s$ to $\cP$ ($b_2$ is in some cluster in $\cC[Q'_{s'}]$). Let $b'_2$ be the closest vertex in in $Q'_{s'}$ to $b_2$ that is in the same cluster with $b_2$. 
By the construction of $H$, $d_H(b_2, b'_2) = 10\e\mu^i$. 
By the inner induction hypothesis, $d_H(b'_2, \pi) \leq f(j + 1) + d_{C'}(b'_2, r') + 11 \cdot s' \cdot \e\mu^i$. 
For any $u_s \in Q'_s$, using the triangle inequality, we get:
\begin{equation*}
\begin{split}
    d_H(u_s, \pi) 
    &\leq d_H(u_s, b'_2) + d_H(b'_2, \pi) \leq d_{H}(u_s, b_2) + d_H(b_2, b'_2) + f(j + 1) + d_{C'} (b'_2, r') + 11\cdot s' \cdot \e\mu^i \\
    &\leq d_{C'}(u_s, b_2) + 10\e\mu^i + f(j + 1) + d_{C'} (b'_2, r') + 11\cdot s' \cdot \e\mu^i \\
    &\leq { \underbrace{d_{C'}(u_s, b_2) + d_{C'}(b_2, r')}_{= d_{C'}(u_s, r')} } + { \underbrace{d_{C'}(b_2, b'_2)}_{\leq \e\mu^i} } + f(j + 1) + 11\cdot s'\cdot\e\mu^i + 10\e\mu^i \\
    &\leq d_{C'}(u_s, r') + f(j + 1) + 11 \cdot s' \cdot \e\mu^i  + 11\e\mu^i\leq d_{C'}(u_s, r') + f(j + 1) + 11 \cdot s \cdot \e\mu^i.
\end{split}
\end{equation*}
The last equation holds since $s' + 1 \leq s$.
This proves the inner induction.
\end{itemize}

\noindent Therefore, for every $s$ and for any $u \in \cC[Q'_s]$, using \Cref{obs:dist_cluster_path} and the inner induction hypothesis:
\begin{equation*}
\begin{split}
    d_H(u, \pi) 
    &\leq d_{H}(r_s, \pi) + 10\e\mu^i \leq d_{C'}(r_s, r') + f(j + 1) + 11\cdot s \cdot\e\mu^i + 10\e\mu^i \\
    &\leq \mu^j + f(j + 1) + 11\cdot \mu^{O(\ddim)} \cdot \e\mu^i \\
    &\leq 2\mu^i + 52\e\mu^i + \sum_{k\in [j+1: i]} \mu^{k}  + (i - j - 1)\cdot \mu^{O(\ddim)} \cdot \e\mu^i + \mu^j + 11\cdot \e^{-O(\ddim)} \cdot \e\mu^i \\
    &\leq 2\mu^i + 52\e\mu^i + \sum_{k\in [j: i]} \mu^{k}  + (i - j)\cdot \mu^{O(\ddim)} \cdot \e\mu^i = f(j)
\end{split}
\end{equation*}
\noindent as desired.  This proves the outer induction. \qed
\end{proof}
\section{Constructing a Constant-Size HPF}
\label{sec:HPF-construct}

Recall the following definition of a $(\rho, \mu)$-hierarchical partition family.
\HPFdef*

We remark that an HPF can either be viewed as a set of hierarchical partitions with an additional covering property (as in the definition above), or as a set of sparse partition families with additional hierarchical structure.
For $K_r$-minor-free graphs, [KLMN'05] (implicitly, see [BFN'19] and [Fil'24]) construct an $(O_r(1), O(1))$-HPF of size $O_r(1)$ with a \emph{weak} diameter guarantee. They achieve this via a black-box reduction from (weak-diameter) $\rho$-PFs: they show that the partitions given by $\rho$-PFs can be massaged into a hierarchy with only an $O(1)$-factor loss in the padding parameter. 
However, their constructed HPF achieves only a weak diameter guarantee on the clusters \emph{even} if one starts from strong-diameter $\rho$-PFs, and it is not clear how any similar argument could preserve the strong diameter.

We fix $\EMPH{$\mu$} \coloneqq \tilde{O}(\ddim^3)$.
The goal of this section is to prove the following lemma:

\HPF*

\subsection{Hierarchies of subnets}
We assume that the minimum pairwise distance in the graph is 1.  
Hence, the aspect ratio $\Phi$ is also the diameter $\Delta$ of the graph; nonetheless, we shall use $\Delta$ rather than $\Phi$ in what follows.  

Given a set of points $S$, a set $N \subseteq S$ is a \EMPH{$t$-covering} of $S$ for some $t > 0$ if for every $v \in S$, there is a point $p \in N$ such that $d_G(v, p)\leq t$. If the minimum distance (in $G$) between any two points in $N$ is at least $t$, $N$ is a \EMPH{$t$-packing}. 
$N$ is a \EMPH{$t$-net} of $S$ if $N$ is a $t$-covering of $S$ and is a $t$-packing. 

A \EMPH{net tree $T$} is induced by a \EMPH{hierarchy of nets} $N_0, N_1, \ldots N_{\log\Delta}$ at levels $0, 1, \ldots, \log \Delta$ of $T$ respectively, satisfying:
\begin{itemize}
\item \textit{[Hierarchical]}\,\, 
$N_{i} \subseteq N_{i-1}$ for every $i$.  In other words, $N_{i-1}$ is a refinement of $N_i$.
\item \textit{[Net]}\,\, 
For every $i$,
$N_i$ is a $\Delta_i$-net of $N_{i-1}$.
\end{itemize}

To build our HPF, we need something stronger than a hierarchy of nets, which provides us the basis for the padding property. The following lemma is the key behind our construction of HPF. 

\begin{lemma}
\label{lem:net-tree-cover}
Fix constants $\rho' \ge 24$
and $\eta \ge 5$ to be determined later. 
Define $\Delta_i \coloneqq \mu^i/(6\eta)$ 
for a sufficiently large parameter $\mu$ to be determined later.
There exists a 
hierarchy $\mathcal{N}$ of nets $N_0, N_1, \ldots N_{\log\Phi}$, where each $N_i$ is a $\Delta_i$-net of $N_{i-1}$
and 
a family of $\sigma$ \EMPH{hierarchies of subnets},
denoted by $\mathcal{N}^1,\ldots,\mathcal{N}^\sigma$, for $\sigma = 2^{O(\ddim)}$, satisfying the following three properties:
\begin{itemize}
\item {\bf Covering:}
For each net $N_i \in \cal{N}$
in the original hierarchy is covered by
$\sigma$ many (not necessarily disjoint) \EMPH{subnets} 
$N^1_i,\ldots,N^\sigma_i$, with
$N^j_i \in \mathcal{N}^j$ for each $j \in [\sigma]$;  $N^1_0,\ldots,N^\sigma_0 = N_0 = V$.
\item {\bf Hierarchical:}
For each $i \ge 1$ and each $j$, 
the subnet
$N^j_i$ provides a
$\mu^i/3$-net of $N^j_{i-1}$,
and a $\frac{5}{12}\mu^i$-cover of $V$.
\item {\bf Padding:} For every $i$ and any vertex $v$, there is a subnet $N_{i}^j$ in some hierarchy $\mathcal{N}^j$,
such that the ball $B(v, \mu^i/\rho')$ around $v$ lies completely within the ball of radius $\mu^i/12$ around some point in $N_{i}^j$.
\end{itemize}

\end{lemma}

\subsubsection {A small family of hierarchies of subnets: Proof of \Cref{lem:net-tree-cover}} \label{subnet_cons}
Our starting point is a hierarchy $\mathcal{N}$ of nets $N_0, \dots, N_{\log\Phi}$, such that $N_i$ is a $\Delta_i$-net of $N_{i-1}$; 
such a hierarchy can be obtained via
a straightforward greedy construction. 

In what follows we employ a two-step procedure.
In the first step, the procedure processes the original hierarchy of nets $\mathcal{N}$ top-down, where
for every $i$, the net $N_i$ is partitioned into $\sigma = \eta^{O(\ddim)}
= 2^{O(\ddim)}$ subsets $\tilde{N}^j_i$ of $N_i$ that are hierarchical, namely, we have
$\tilde{N}^j_i \subseteq \tilde{N}^j_{i-1}$ for each $j$ and any $i \ge 1$.
In the second step, the resulting $\sigma$ hierarchies of subsets 
\smash{$\{\tilde{N}^j_i\}$}
are processed bottom-up, to get   $\sigma$ many new hierarchies of subnets \smash{$\{N^j_i\}$}.

\paragraph{Step 1.}
We use the following iterative process, as $i$ decreases from $\log\Phi$ down to $0$.  
We handle the case $i = \log \Phi$ separately. 
Noting that $N_{\log\Phi}$ contains a single point, denoted by $r$, we define $\tilde{N}^j_{\log\Phi} = 
N_{\log\Phi} = \{r\}$, for every $j \in [\sigma]$. 

Next consider $i < \log\Phi$. To build $\tilde{N}^j_i$, for each $j \in [\sigma]$,  
   we first place all points of $\tilde{N}^j_{i+1}$ in $\tilde{N}^j_i$.
    Once $\tilde{N}^j_i$ has been built in this way for each $j$, 
    consider the remaining points in $N_i$ that are not covered by any of the level-$(i+1)$ 
    subsets $\tilde{N}^j_{i+1}$; that is, consider the points in
    \[
    N'_i \coloneqq N_i \setminus \Paren{ \bigcup_j \tilde{N}_{i+1}^j }.
    \]
    Now for each $j$ from  1 to $\sigma$, we complete $\tilde{N}_i^j$ by adding to it greedily the points in $N'_i \setminus (\bigcup_{k<j} \tilde{N}_{i}^{k})$ using the ball-carving algorithm  with a packing parameter $\mu^i/3$. 
   If $i >0$, decrement $i$ and repeat.

\paragraph{Step 2.}
We use the following iterative process, as $i$ increases from 0 to $\log\Phi$.
We handle the case $i = 0$ separately, by defining $N^j_{0} = N_0 = V$ 
for every $j \in [\sigma]$. 

Next, consider $i > 0$.
For each $j \in [\sigma]$, we complete $\tilde{N}_i^j$ (which is a $\mu^i/3$-packing) 
into a $\mu^i/3$-net for ${N}_{i-1}^j$, denoted by
${N}_i^j$,
by adding to it greedily the points in ${N}_{i-1}^j$ using the ball-carving algorithm with a packing parameter $\mu^i/3$. 
(Note that the subnets ${N}_i^j$ are not necessarily contained in $N_i$.)
If $i < \log\Phi$, increment $i$ and repeat.

\paragraph{Analysis.}
We establish the three properties required by \Cref{lem:net-tree-cover} individually.

\begin{itemize}
\item {\bf Covering.~}
By construction, for every $i$, the subsets $\tilde{N}^j_i$ of $N_i$ are pairwise disjoint. 
Since the ball of radius $\mu^i/3$ around any point of $N_i$ contains less than $\sigma$ points of $N_i$, at the end of Step~1 of the procedure, any point of $N_i$ will be added to a subset $\tilde{N}_i^j$, for each $i$; thus $N_i$ is covered by the $\sigma$ subnets $\tilde{N}^j_i$.
Since $N^j_i$ is a superset of $\tilde{N}^j_i$, for every $i$ and $j$, we have
$\bigcup_j N^j_i \supseteq N_i$, i.e., $N_i$ is covered by the $\sigma$ subnets ${N}^j_i$.
Finally, by construction we have $N^1_0,\ldots,N^\sigma_0 = N_0 = V$.

\item {\bf Hierarchical.~}
Recall that $N_0^j = N_0 = V$ for each $j$.
By construction, for every $i \ge 1$ and $j$: \begin{enumerate} \item 
${N}^j_i \subseteq
\tilde{N}^j_i \cup 
{N}^j_{i-1} \subseteq
{N}^j_{i-1}$, hence the subnets $N_i^j$ satisfy the hierarchical property.
\item ${N}^j_i$ is a $\mu^i/3$-net of ${N}^j_{i-1}$. 
Thus, for each $j$, $N^j_1$ is a $\mu/3$-covering of $V$, $N^j_2$ is a $\mu^2/3$-covering of $N^j_1$, and thus a $(\mu^2 + \mu^1)/3$-covering of $V$.
Inductively we get that each $N^j_i$ provides an $s_i$-covering for $V$,
where $s_i = \frac{1}{3}\left(\sum_{k = 1}^i \mu^k\right) < 
\frac{5}{12} \mu^i$, where the second inequality holds for a sufficiently large $\mu$; this upper bound holds for every $i \ge 1$ and $j$ (for $i= 0$ it holds trivially).
\end{enumerate}

\item \textbf{Padding.~}
The padding property holds trivially for $i= 0$ and any vertex $v$, since $N^1_0,\ldots,N^\sigma_0 = N_0 = V$.
We henceforth fix an arbitrary $i \ge 1$ and an arbitrary vertex $v$.  Since $N_i$ is a $\Delta_i$-net for $N_{i-1}$, for each $i \ge 1$, and as $N_0 = V$,
we inductively get that $N_i$ is a 
$s'_i$-covering for $V$, where 
\[
s'_i 
~=~ 
\sum_{k=1}^i \Delta_k 
~=~ 
\frac{\sum_{k=1}^i  
\mu^k}{6\eta}
~<~ 
\frac{\mu^{i+1} - 1}{
\mu -1} \cdot \frac{1}{6\eta}~<~
\frac{5}{4} \mu^i \cdot \frac{1}{6\eta},
\]
where the second inequality holds for a sufficiently large $\mu$.
Hence, there is a net point $p$ in $N_i$ at distance at most $s'_i$ from $v$, and so
$d_G(p,v) \le
s'_i < 
 \frac{5}{4} \mu^i/(6\eta) \le \mu^i / 24$, where the last inequality holds for $\eta \ge 5$.
By setting $\rho' \ge 24$, the 
ball $B(v, \mu^i/\rho')$ around $v$ is strictly contained in the ball of radius $\mu^i/12$ around $p$. 
Finally, as $N_i$ is covered by the $\sigma$ subnets ${N}^j_i$,
$p$ belongs to at least one subnet $N_i^j$.
\end{itemize}

This completes the proof of \Cref{lem:net-tree-cover}.

\subsection{Algorithm to build HPF using \Cref{lem:net-tree-cover}}
\label{ssec:alg-hpf-net-tree}
We will build the HPF from the ground up, where we construct a separate $\mu$-HP $\mathbb{P}^j$ for each of the $\sigma =2^{O(\ddim)}$ hierarchies  $\mathcal{N}^j$ of subnets, $j \in [\sigma]$, provided by \Cref{lem:net-tree-cover}. 
(While the diameter bound of the clusters to be built will be three times the lower bound on the distance between subnet points, the growth parameter is still $\mu$.)
In addition, we will show the \emph{padding property}: at each level $i$, for any ball $B(v,\mu^i/\rho)$, one of the partitions from the $\mu$-HP $\mathbb{P}^j$, for some $j \in [\sigma]$, has a cluster containing the ball.

Our goal is to build a partition $P^j_i$ of clusters with strong diameter $\mu^i$ associated with the points in the subnets $N_i^j$, for every $i \in [\log\Phi]$ and $j \in [\sigma]$, 
where $\sigma = 2^{O(\ddim)}$. 
For the rest of the algorithm, fix any $j \in [\sigma]$ and the hierarchy $\mathcal{N}^j$.
For the base case $i= 0$, $P^j_0$ consists of all vertices of $V$ as singleton clusters,
each of which of strong diameter $0 \le \mu^0$. We henceforth consider the case $i > 0$.

Assume the partition $P_{i-1}^j$ into level-$(i-1)$ clusters of strong diameter $\Delta' \coloneqq \mu^{i-1}$ centered at points in $N_{i-1}^j$ (and which satisfies the padding property) 
has been built; next, we will build the partition $P^j_i$ into level-$i$ clusters of strong diameter $\mu^i$ and make sure that the padding property is satisfied.  
We do so using the \EMPH{cluster aggregation} problem; an instance of this problem can be described as follows:
\begin{itemize}
    \item \textit{Input:}\,\, 
    An edge-weighted graph with vertex partition into clusters $\mathcal{C}$ of strong diameter $\Delta' \coloneqq \mu^{i-1}$ where the cluster centers form a {\em relaxed $\mu^{i-1}/3$-net}, 
    meaning that they form 
    a $\mu^{i-1}/3$-packing and a $\frac{5}{12}\mu^{i-1}$-covering of $V$,
    and a set of portals $N \coloneqq N_i^j$ that is a relaxed $\mu^i/3$-net, i.e.,  
    $N^j_i$ is 
    a $\mu^i$-packing and a $\frac{5}{12}\mu^{i}$-covering of $G$.
    \item \textit{Output:}\,\,  
    An assignment function $f: V \to N$ that maps clusters to portals, such that 
    (1) every vertex from the same cluster maps to the same portal under $f$ (in other words, $f(C)$ is well-defined for any cluster $C$), 
    (2) for every vertex $v$, the distance in the subgraph induced by the respective cluster of $V$
    between $v$ and its assigned portal $f(v)$
    is at most an additive $O(\Delta')$ distortion from the graph distance between $v$ and its closest portal in $N$; that is,
    \[
        d_{G[f^{-1}(f(v))]} (v, f(v)) \le d_G(v, N) + \beta\Delta'.
    \]
\end{itemize}
We carry out the construction of HPF assuming the existence of an additive $\beta\Delta'$ solution to the cluster aggregation problem~\cite{BCF+23}. Here, $\beta = \tilde{\Theta}(\ddim^2)$. We choose $\mu = \Theta(\ddim^3) \geq 12\beta + 12.$
Specifically, we perform cluster aggregation on the partition
$P_{i-1}^j$ into level-$(i-1)$ clusters with $N = N_i^j$ being the set of portals; 
by \Cref{lem:net-tree-cover}, the input to the cluster aggregation problem is valid.
Let $f$ be the resulting assignment function, and treat $\set{f^{-1}(p) : p \in N}$ 
as the set of level-$i$ clusters, to create a new cluster partition $P_{i}^j$.  

Since there are $\sigma = 2^{O(\ddim)}$ different hierarchies $\mathcal{N}^j$, we will have $2^{O(\ddim)}$ partitions in the final HPF.
We next show that the resulting HPF satisfies the required conditions.

\begin{claim}
    \label{clm:strg-diam}
    Level-$i$ clusters have strong diameter $\mu^i$.
\end{claim}
\begin{proof}
By the guarantee of cluster aggregation, each vertex $v$ in any level-$i$ cluster is at most $d_G(v, N_{i}^j) + \beta\Delta'$ away from the unique portal within.   
Since $N_{i}^j$ is a relaxed $\mu^i/3$-net, it  provides a $\frac{5}{12} \mu^i$-covering for $V$, and so $d_G(v, N_{i}^j) \le \frac{5}{12} \mu^{i}$. Hence each level-$i$ cluster has (strong) diameter at most 
$2 \cdot (d_G(v, N_{i}^j) + \beta\Delta') \le 2(\frac{5}{12}\mu^{i} + \beta\mu^{i-1})$,
thus strong diameter $\mu^i$ because  $\mu \ge 12\beta + 12$. 
\end{proof}

\begin{claim} \label{cl:unique}
Any level-$(i-1)$ cluster intersecting $B(p,\mu^i/12)$ is assigned to the portal $p$.
\end{claim}
\begin{proof}
Let $C$ be a cluster of diameter $\Delta' = \mu^{i-1}$ intersecting the ball $B(p,\mu^i/12)$.
We argue that 
$C$ must be assigned to portal $p$ by the cluster aggregation procedure.
Assume for contradiction that $C$ is assigned to another portal $p'$.  Since the portals belong to $N_i^j$, 
which is a $\mu^i/3$-packing, we have that 
$d_G(p,p') > \mu^i/3$.
Cluster $C$ has diameter $\mu^{i-1}$ and it intersects the ball $B(p,\mu^i/12)$, hence any point inside $C$ is within distance at most $\mu^i/12 + \mu^{i-1}$ from $p$ (and thus from $N_i^j$).
However, the assumption that $C$ is assigned to $p'$ implies that, for any $v \in C$,
\[
    d_G (v,p') 
    ~\le~ d_{G[f^{-1}(p')]} (v, p') 
    ~\le~ d_G(v, N_i^j) + \beta\Delta' 
    ~\le~ \mu^i/12 + \mu^{i-1} + \beta\mu^{i-1}.
\]
Triangle inequality then derives a contradiction, 
since $\mu \geq 12\beta + 12$: 
\[
    d_G (p,p') ~\le~ d_G (p,v) + d_G (v,p')
    ~\le~ (\mu^i/12 + \mu^{i-1}) + (\mu^i/12 + \mu^{i-1} + \beta\mu^{i-1})
    ~<~ \mu^i.
\]
\aftermath
\end{proof}

\begin{claim}
\label{clm:HPF-correctness}
The partitions $\mathbb{P}^j$ produce a $(\rho,\mu)$-HPF.
\end{claim}
\begin{proof}
By the construction and the previous two claims, we have that among the $\sigma = 2^{O(\ddim)}$ hierarchical partitions $\mathbb{P}^j = (P^j_0, \dots, P^j_{\log\Phi})$ generated, $j \in [\sigma]$, each level-$i$ cluster in $P^j_{i}$ has strong diameter at most $\mu^i$ and it is the union of some level-$(i-1)$ clusters in $P^j_{i-1}$. 
It remains to show the \emph{padding property}: for every level $i$ and for every vertex $v$, there is a scale-$i$ partition $P^j_i$, for some $j \in [\sigma]$, and some cluster $C \in P^j_i$, such that the ball $B(v, \mu^i / \rho)$ is contained in $C$.

Fix a level $i$ and an arbitrary vertex $v$.
Choose $\rho = \rho' \ge 24$.
By Lemma~\ref{lem:net-tree-cover}, there is a subnet $N_{i}^j$ in the hierarchy $\mathcal{N}^j$, for some $j \in [\sigma]$, such that the ball $B(v, \mu^i/\rho)$ around $v$ lies completely within the ball of radius $\mu^i/12$ around some point $p$ in $N_{i}^j$.
Consider the level-$i$ cluster $C_p$ corresponding to portal $p$.
\Cref{cl:unique} implies that all level-$(i-1)$ clusters intersecting $B(p, \mu^i/12)$ are merged into cluster $C_p$.
As a result, the ball $B(v, \mu^i/\rho)$ lies completely within $C_p$; this proves the padding property.
\end{proof}

For our application to spanning tree cover, we will need the fact that every cluster in a level of any HP is partitioned into constant number of cluster in the next level. 
We then need to bound the \emph{maximum degree} of an HPF.

In our construction of $\mathbb{P}^j$ for each $j$, recall that for each cluster $C$ built from a net point $u$, we set $u$ to be the representative of $C$.

\begin{observation}
    \label{obs:bdd-deg}
    The partition $\mathbb{P}^j$ has maximum degree $\mu^{O(\ddim)}$.
\end{observation}

\begin{proof}
    Let $\mathbb{P}^j = \Set{P^j_0, P^j_1, P^j_2, \ldots P^j_{\log\Delta} }$. By \Cref{clm:strg-diam}, each cluster $C$ in $P^j_i$ has strong diameter $\mu^i$. 
    Since the cluster centers of children 
    of $C$ are subset of a $\mu^{i - 1}/3$-net, using packing bound, we get the total number of children of $C$ is at most $\left(\frac{3\mu^i}{\mu^{i - 1}}\right)^{\ddim} = \mu^{O(\ddim)}$. 
\end{proof}

\smallskip
\noindent Therefore \Cref{lem:HPF} follows from \Cref{clm:strg-diam}, \Cref{cl:unique} and \Cref{obs:bdd-deg}. 

\newcommand{\dfs}[1]{\ensuremath{t_{#1}}}

\section{Optimal Routing in Doubling Graphs}
\label{S:routing}

A \EMPH{routing scheme} is a distributed algorithm for routing (or sending)  packet headers (or messages) 
from any source to any destination vertex in a network. 
We restrict attention to {\em labeled routing schemes}, where the network is preprocessed to assign each vertex a unique \EMPH{label} (also called an address) and a \EMPH{local routing table}; in contrast, in {\em name-independent routing schemes} vertex labels are chosen adversarially (and are independent of the routing scheme).  
The edges incident to each vertex are given \EMPH{port numbers}; in the \EMPH{designer-port} model these port numbers can be chosen by the algorithm designer in preprocessing, whereas in the \EMPH{fixed-port} model these port numbers are chosen adversarially.
\footnote{Normally, port numbers around a vertex $v$ must lie in the range $[1, \deg(v)]$. We relax this assumption and only require that the port numbers are $O(\log n)$-bit numbers; that is, our results hold even in a more demanding setting than the standard setting. We do this to simplify the presentation.}
An instance of routing begins at an arbitrary source vertex, which sets up a packet header (a message), based on the labels of the source and destination vertices, and based on 
the local routing table of the source vertex.
Upon the reception of a message at a vertex, that vertex decides whether the message has reached its destination and, if not, where to forward it, based on the message it received and its local routing table; in the latter case, that vertex selects a port number and forwards the message along the corresponding edge.
The routing scheme is said to have \EMPH{stretch} $k$ if the path found by any instance of routing, from any source $s$ to any destination $t$, has length at most $k$ times the distance between $s$ and $t$ in the graph. 
The basic goal is to achieve a low stretch, ideally approaching 1, via a {\em compact} routing scheme, which means that
the routing table size at each vertex, as well as the size of all labels and messages, should be polylogarithmic in the network size, and ideally bounded by $O(\log n)$ bits.
Additionally we would like to bound the \EMPH{routing decision time}, which is the time taken (assuming a word RAM model) by a node to compute the port to which it will forward its message; ideally the decision time should be constant.
We use our spanning tree cover to achieve such a result.

\RoutThm*

The problem of labeled routing on \emph{trees} has been well studied, and speaking roughly, routing on a single tree is much easier than on other graph classes. This is why our spanning tree cover will be useful. To construct a routing scheme for $G$, we want to construct a tree cover for $G$ and then route along one of the trees from the tree cover. However, in general trees, fixed-port routing requires labels or routing tables of size $\Theta(\log^2 n/ \log \log n)$ bits: this is tight if one demands exact routing (stretch $1$), and nothing better is known even if one allows $1+\e$ approximation.
We show that we can bypass these lower bounds for the trees that we care about (coming from our tree cover). In order to do this, we first preprocess graph $G$ by constructing a greedy spanner\footnote{that is, iterate over all edges in order of increasing weight, and add the edge $(x,y)$ to the spanner $G'$ unless $G'$ already satisfies $\dist_{G'}(x,y) \le (1+\e) \dist_G(x,y)$.} $G'$ for $G$, and then construct a spanning tree cover for $G'$. We show that the resulting trees have very compact routing schemes. This is the focus of \Cref{ssec:routing-spanning-tree}.
\begin{lemma}
\label{lem:routing-spanning-tree}
    Let $\e \in (0,1)$. Let $G$ be a graph with doubling dimension $\ddim$, and let $G'$ be a $(1+\e)$-approximate greedy spanner for $G$. Let $T$ be an $n$-vertex tree which is a subgraph of $G'$. Then there is a $(1+\e)$-stretch labeled routing scheme for $T$ in the fixed-port model where the sizes of labels, headers, and routing tables are $\e^{-O(\ddim)} \cdot \log n$ bits. The routing decision time is $\e^{-O(\ddim)}$.
\end{lemma}

The above lemma is not sufficient to prove \Cref{thm:routing} --- given a source vertex $x$ and vertex $y$, before routing along a tree $T$ from our tree cover, we first need to find an appropriate tree $T$ that approximately preserves distance between $x$ and $y$. This is the focus of \Cref{ssec:find-correct-tree}.
\begin{restatable}{lemma}{findCorrectTree}
\label{lem:find-correct-tree}
    Let $\e \in (0,1)$. Let $\cT = \set{T_1, \ldots, T_k}$ be the spanning tree cover for $G$ constructed by \Cref{thm:spanning-main}, where $k = \e^{-\tilde O(\ddim)}$. There is way to assign $\e^{-\tilde O(\ddim)} \cdot \log n$-bit labels to each vertex in $V(G)$ so that, given the labels of two vertices $x$ and $y$, we can identify an index $i$ such that tree $T_i$ is a ``distance-approximating tree'' for $u$ and $v$; that is, 
    $\dist_{T_i}(x,y) \le (1+\e) \cdot \dist_G(x,y)$. This decoding can be done in $O(d \cdot \log 1/\e)$ time.
\end{restatable}

Equipped with these lemmas, we now describe our routing scheme to prove \Cref{thm:routing}. We first construct a $(1+\e)$-approximate greedy spanner
$G'$ for $G$.
Next construct a $(1+\e)$-spanning tree cover $\cT$ on $G'$, as guaranteed by \Cref{thm:spanning-main}. Observe that $G'$ has doubling dimension $O(\ddim)$, as distances are only distorted a $1+\e$ factor.%
\footnote{A radius-$r$ ball in $G'$ contains every vertex in a radius-$\frac{r}{1+\e}$ ball in $G$. Thus packing bound implies that a radius-$r$ ball in $G$ can be covered with $2^\ddim$ many radius-$r$ balls in $G'$. Now consider a radius-$2r$ ball in $G'$. Such a ball is covered by a single radius-$2r$ ball in $G$, so it can be covered by $2^\ddim$ radius-$r$ balls in $G$, so it can be covered by $2^\ddim \cdot 2^\ddim = 2^{O(\ddim)}$ radius-$r$ balls in $G'$.}
For each of the $\e^{-\tilde O(\ddim)}$ trees $T \in \cT$, we construct a fixed-port labeled routing scheme on $T$ using \Cref{lem:routing-spanning-tree}, where (for each instance of routing on a tree $T$) we set the port numbers around a vertex $x$ in $T$ to be the same as the provided port numbers around $x$ in $G$.\footnote{Note that the port number for $x$ in tree $T$ will not necessarily be in the range $[1, \deg_T(x)]$, but rather could be in the range $[1, \deg_G(x)]$; nevertheless each port number is $O(\log n)$ bits.}
The final routing table of a vertex $x$ in $V(G)$ stores the concatenation of the routing tables of $x$ that correspond to each of the trees in $\cT$.
The final label of $x$ is given as the concatenation of the labels of $x$ that correspond to each of the trees in $\cT$, as well as an additional label component from \Cref{lem:find-correct-tree} which lets us identify a distance-approximating tree.
While routing in $G$, we first identify a distance-approximating tree $T$ using the labels of the source and destination, and then route along $T$ using the stored labels/routing tables/port numbers for $T$.

The stretch of our routing scheme is $(1+\e)^3 \le 1+O(\e)$, because the the stretch of the greedy spanner is $(1+\e)$, the stretch of the tree cover on the spanner is $(1+\e)$, and we use a $(1+\e)$-stretch routing scheme on these trees.
The routing scheme uses $\e^{-\tilde O(\ddim)} \cdot \log n$ bit routing tables and labels.
The routing decision time is dominated by the time it takes to route along $T$, which is $\e^{-O(\ddim)}$.
Rescaling $\e \gets \e/O(1)$ proves~\Cref{thm:routing}.

\subsection{Routing along a single tree}
\label{ssec:routing-spanning-tree}
This section is dedicated to the proof of \Cref{lem:routing-spanning-tree}. Let \EMPH{$\e$} be some value in $(0,1)$. Let $G$ be a graph with doubling dimension $\ddim$, and let $G'$ be a $(1+\e/3)$-approximate greedy spanner of $G$. Let \EMPH{$T$} be a tree which is a subgraph of a $G'$.

\paragraph{Interval routing for bounded-degree trees.} We first describe a classic stretch-1 (i.e., exact) routing scheme of \cite{SK85} which works well for trees of small degree. Fix an arbitrary root and perform a DFS traversal over $T$. The \EMPH{DFS timestamp} of a vertex $x$, denoted \EMPH{$\dfs x$}, is the time $x$ was first reached in the traversal. The \EMPH{DFS interval} of each vertex $x$ is the closed interval $[a,b]$ where $a \coloneqq \dfs x$ and $b$ is the largest DFS timestamp of any descendant of $x$. Observe that a vertex $y$ is a descendant of $x$ if and only if the $\dfs y$ lies within the DFS interval of $x$.
Moreover, the DFS intervals of the \emph{children} of $x$ (together with $\dfs x$) form a partition of the DFS interval of $x$. In our routing scheme, the label of vertex $x$ is simply $\dfs x$. The local routing table of $x$ stores the DFS interval of $x$ and the port number of the parent of $x$, as well as the DFS intervals and port numbers of each child of $x$. To route to a target vertex $y$ from a current vertex $x$, we check the local routing table of $x$: if $\dfs y = \dfs x$ then we are done; otherwise if $\dfs y$ is outside the DFS interval of $x$ then we route to the parent of $x$; and otherwise $\dfs y$ lies in some DFS interval of a child of $x$ and we route along the corresponding port number.
The size of the label of a vertex $x$ is $O(\log n)$ bits. The size of the routing table of $x$ is $O(\deg(x) \cdot \log n)$ bits, as we need $O(\log n)$ bits per child to store the associated interval and port number.

\paragraph{Almost bounded-degree property.} In our setting, the tree $T$ might \emph{not} have bounded degree. However, we can achieve something similar if we restrict ourselves to looking at edges of roughly the same weight. Here is where we use the assumption that $T$ is a subgraph of a greedy spanner of a graph with bounded doubling dimension.
\begin{claim}
\label{clm:almost-degree}
    There is some $\EMPH{$\alpha$} = \e^{-O(\ddim)}$ such that: For any $\ell > 0$ and any vertex $x$ in $T$, there are at most $\alpha$ edges of $T$ incident to $x$ with weight in the range $[\ell, 2\ell]$.
\end{claim}
\begin{proof}
    We prove a stronger statement: the same bound holds in the greedy spanner $G'$ (and thus it holds is $T$, which is a subgraph of $G'$). For any vertex $v \in V(G')$, let \EMPH{$N_v$} denote the set of vertices in $V(G')$ that are connected to $v$ in $G'$ by an edge with weight in $[\ell, 2\ell]$. For every pair of vertices $a,b \in N_v$, we have $\dist_G(a,b) > \e \cdot 2 \ell$ --- otherwise, the $(1+\e)$-approximate greedy spanner $G'$ would not include both edges $(v,a)$ and $(v,b)$. This means that any ball of radius $\e \ell$ in $G$ covers only one vertex in $N_v$. On the other hand, every vertex of $N_v$ is covered by the ball of radius $2 \ell$ centered at $v$. Packing bound now implies that $|N_v| \le \e^{-O(\ddim)}$.
\end{proof}

\paragraph{Define labels and routing tables.} We are now ready to construct labels and routing tables for the vertices of $G$. As in the interval routing scheme, consider a DFS traversal of $T$ from an arbitrary root vertex. During this DFS traversal, always take the \EMPH{minimum-weight} edge possible. (We note the following consequence: for any vertex $x$, if we let $(c_1, c_2, \ldots, c_{k})$ denote the sequence of children of $v$ sorted in increasing order of DFS timestamp, we have that the weights of the corresponding edges are also in increasing order, i.e. $w((x, c_1)) \le \ldots \le w((x, c_k))$.)
The \EMPH{label} of vertex $x$ is $\dfs x$. Let \EMPH{$\beta = 2 \log (1/\e) \cdot \alpha$}, where $\alpha$ is the value from \Cref{clm:almost-degree}; the choice of $\beta$ will be made clear in the analysis. The \EMPH{routing table} of $x$ consists of the following four items:
\begin{enumerate}
    \item Store the DFS interval of $x$, and the port number of the parent of $x$.
    \item Let \EMPH{$c_1, \ldots, c_\beta$} denote the $\beta$ children
    of $x$ with the \emph{smallest} DFS timestamps. Store (in order) the DFS interval and port number corresponding to each of these $\beta$ \EMPH{children} vertices.
    \item Let $x'$ be the parent of $x$. Store the DFS interval of $x'$.
    \item Let $c'_1, c'_2, \ldots, c'_k$ be the children of $x'$ (i.e. the siblings of $x$), sorted in increasing order of DFS timestamps. Let $i$ be the index such that $x = c'_i$. Consider each of the vertices $c'_{i+1}, \ldots, c'_{i + \beta}$: store (in order) the DFS interval and the port number \emph{of $x'$} corresponding to each of these $\beta$ \EMPH{sibling} vertices.
\end{enumerate}

\paragraph{Define routing algorithm.} Given the label of a destination vertex $y$, an additional message header, and the routing table of the current vertex $x$, we now describe how to route to the next vertex. The message header will either be empty or will contain a port number. Let \EMPH{$x'$} denote the parent of $x$.
\begin{itemize}
    \item \textbf{Case 0.} Use (Item 1) of the routing table to check if $\dfs y = \dfs x$. If so, we are done.
    \item \textbf{Case 1.} Use (Item 1) of the routing table to check if $\dfs y$ is inside the DFS interval of $x$. If so:
    \begin{itemize}
        \item \textbf{Case 1a.} Check the children $(c_1, \ldots, c_\beta)$ stored by (Item 2). If there is some child $c_i$ whose DFS interval contains $\dfs y$, route to that child, sending an empty header.
        \item \textbf{Case 1b.} Otherwise, if the header contains a port number, route to that port. Send an empty header.
        \item \textbf{Case 1c.} Otherwise, route to the smallest child $c_1$, sending an empty header.
    \end{itemize}

    \item \textbf{Case 2.} Otherwise, $\dfs y$ is outside the DFS interval of $x$.
    \begin{itemize}
        \item \textbf{Case 2a.} Use (Item 3) to check if $\dfs y$ is outside the DFS interval of $x'$. If so, route to $x'$ (using the stored port number from Item 1) and send an empty header.
        \item \textbf{Case 2b.} Otherwise, use (Item 4) to check if $\dfs y$ is in any of the DFS intervals of the $\beta$ stored siblings $(c'_{i+1}, \ldots, c'_{i+\beta})$ of $x$. If so, $\dfs y$ is is in the DFS interval of some sibling $c'_z$, and we try to eventually route to $c'_z$. To do this, we route to the parent $x'$ and send a header with the stored port number (for $x'$) that corresponds to $c'_z$.
        \item \textbf{Case 2c.} Otherwise, we try to eventually route to the sibling $c'_{i+1}$. To do this, we route to the parent $x'$ and send a header with the stored port number (for $x'$) that corresponds to $c'_{i+1}$.
    \end{itemize}
\end{itemize}

\paragraph{Analysis.}
The size bound is simple. Let $x$ be a vertex. The label of $x$ is $O(\log n)$ bits. For the routing table of $x$, we need $O(\log n)$ bits to store the DFS intervals and port numbers for $x$, the parent of $x'$, $\beta$ children of $x$, and $\beta$ siblings of $x$; in total we need $O(\beta \log n) = \e^{-O(\ddim)} \log n$ bits. The message header consists of (at most) a single port number, so it is also $O(\log n)$ bits. The routing decision time is bounded by $\e^{-O(\ddim)}$, as the algorithm simply reads a sequence of $\e^{-O(\ddim)}$ words and tests whether a certain number is inside or outside of certain intervals. It remains to show the stretch bound.

\begin{claim}
    The routing scheme described above has stretch $1+\e$.
\end{claim}
\begin{proof}
    The proof is by induction: if we start at some vertex $x$ and want to route to vertex $y$, we show there is some vertex $v$ on the path in $T$ between $x$ and $y$ such that we eventually route to $x$, after traveling only $(1+\e) \cdot \dist_T(x,v)$ distance.
    If $x = y$, we are done (Case 0). If $y$ is not a descendant of $x$, then $\dfs y$ lies outside the DFS interval of $x$, and we immediately route to the parent of $x$ (which lies on the shortest path between $x$ and $y$); this occurs in every subcase of Case 2. By induction, we are done. 
    
    The interesting case is when $y$ is a proper descendant of $x$. Let $(c_1, \ldots, c_k)$ denote the children of $x$, sorted in increasing order of DFS timestamps. Let \EMPH{$z$} denote the index such that $y$ is a descendant of the child \EMPH{$c_z$}.
    If $z \in [1,\beta]$, then we route to $c_z$ immediately and, by induction, we are done. Otherwise, we begin by routing to $c_1$. After routing to $c_1$, we will be in Case 2b or Case 2c (because $\dfs y$ belongs is in the DFS interval of a sibling of $c_1$). If $z \in [2, \beta+1]$, then we are in Case 2b; we route to $x'$ with a nonempty header and then immediately route to $c_z$ (by Case 1b on $x$). Otherwise, we are in Case 2c; we route to $x'$ with a nonempty header and then immediately route to $c_3$ (by Case 1b on $x$). In general, we continue traveling through the children $c_i$ in increasing order from $i \gets 1$ up until we reach the index $i \gets z-\beta$, after which we route to $c_z$.
    (Note that $z-\beta > 1$, as otherwise we would have immediately routed to $c_z$.) 
    It remains to analyze the total distance traveled before we reach $c_z$.

    For every child $c_i$ of $x$, let $\ell_i$ denote the weight of the edge $(x, c_i)$. The total distance traveled before reaching $c_z$ is
    \[\ell_z + \sum_{i=1}^{z - \beta} 2 \ell_i \]
    because we travel along each edge $(x, c_i)$ twice (once to reach $c_i$ for the first time, and once to route back to $x$ on the way to $c_{i+1}$) before eventually routing along $(x, c_z)$. 
    By \Cref{clm:almost-degree}, we have $\ell_i \ge \ell_{i-\alpha}/2$. By grouping together the terms of the sum $\sum_{i=1}^{z - \beta} 2 \ell_i$ into chunks of size $\alpha$, we arrive at a geometric series where the dominant term is $\alpha \cdot 2 \ell_{z-\beta}$. Thus,
    \[\sum_{i=1}^{z - \beta} 2 \ell_i \le 4 \cdot \alpha \cdot \ell_{z - \beta}.\]
    By definition, $\beta = 2 \log(1/\e) \cdot \alpha$. \Cref{clm:almost-degree} now implies
    \[\ell_{z-\beta} \le \ell_z / 2^{2 \log (1/\e)} = \e \cdot \ell_z / 4.\]
    We conclude $\sum_{i=1}^{z - \beta} 2 \ell_i \le \e \cdot \ell_z$, meaning that we travel a total distance of $(1+\e)\ell_z = (1+\e) \dist_T(x, c_z)$ before reaching $c_z$, as desired.
\end{proof}

\subsection{Finding the correct tree}
\label{ssec:find-correct-tree}
The goal of this section is to prove \Cref{lem:find-correct-tree}, which constructs labels that let us find an appropriate tree in our tree cover for routing. 
We restate it here for convenience.

\findCorrectTree*

The approach is nearly identical to \cite{CCL+24b}, who proved something similar for their Euclidean tree cover; we need to adapt some details to our specific tree cover construction. We begin by reducing the problem to a variant of lowest-common-ancestor labeling. We need to introduce some notation. Recall that the spanning tree cover of \Cref{thm:spanning-main} was constructed by starting with a pair-preserving $(\mu, \rho)$-HPF $\mathfrak{H} = \set{\mathbb{H}_1, \ldots, \mathbb{H}_\ell}$; for each $j \in [1, \log_\mu(1/\e)]$ and $\mathbb{H} \in \mathfrak{H}$, we constructed a tree $T^j_{\mathbb{H}}$. 
We now explicitly define an object that was implicitly used to construct the tree $T^i_{\mathbb H}$. For every hierarchy $\mathbb H \in \mathfrak H$ and $j \in [1, \log_\mu(1/\e)]$, we define the \EMPH{subhierarchy $\mathbb H^j$} to be the hierarchy obtained by starting with the sequence of partitions $\mathbb H = \set{P_1, \ldots, P_{i_{\max}}}$ and then \emph{coarsening} $\mathbb H$ by deleting from $\mathbb H$ every partition  $P_i$ with $i \not \equiv j \mod \log_\mu (1/\e)$.
Recall that any hierarchy can be viewed alternatively as a tree of clusters; viewing $\mathbb H^j$ under this perspective,
observe that for any cluster $C$ in $\mathbb H^j$, the children
of $C$ in $\mathbb H^j$ are precisely the $\e$-subclusters of $C$ in $\mathbb H$, and the leaves of $\mathbb H^j$ are singleton clusters (each containing a single vertex). Each cluster $C$ in $\mathbb H$ has at most $\e^{-O(\ddim)}$ $\e$-subclusters\footnote{this follows from the degree bound of \Cref{lem:HPF}; see the proof of \Cref{lm:pair_preserve_HPF}} and so $C$ has at most $\e^{-O(\ddim)}$ children in $\mathbb H^j$.
The spanning tree cover algorithm implicitly constructs the subhierarchy $\mathbb H^j$ while constructing the tree $T^j_{\mathbb H}$ (cf. \Cref{obs:subtree}).

For every $\mathbb H \in \mathfrak H$ and $j \in [1, \log_\mu(1/\e)]$, we define the \EMPH{compressed subhierarchy $\check{ \mathbb H}^j$} to be a tree obtained by viewing $\mathbb H^j$ as a tree of clusters, and contracting away all clusters that have exactly one child. 
If cluster $C$ in $\mathbb H$ is labeled with an $\e$-subcluster pair $(C_1, C_2)$, then we say that the corresponding cluster in $\check{\mathbb H}^j$ (if it was not deleted or contracted away) is also labeled with the same pair $(C_1, C_2)$; note that $C_1$ and $C_2$ are children of $C$ in $\check{\mathbb H}^j$. If $G$ has $n$ vertices, then $\check {\mathbb H}^j$ contains only $O(n)$ clusters. Let $\check {\mathfrak H}$ be the set of compressed subhierarchies $\check{\mathbb H}^j$ for all $\mathbb H \in \mathfrak H$ and $j \in [1, \log_\mu(1/\e)]$.

Given a compressed subhierarchy $\check {\mathbb{H}}^j$ and a pair of vertices $x, y \in V(G)$, we abuse notation and let \EMPH{$\lca(u,v)$} denote the cluster in $\check {\mathbb{H}}^j$ which is the lowest common ancestor (for short, LCA) of the leaves corresponding to $x$ and $y$ in $\check {\mathbb{H}}^j$. We say that $\check {\mathbb{H}}^j$ satisfies the \EMPH{LCA condition for $(x,y)$} if
$\lca(x,y)$ is labeled with a pair of subclusters $\set{C_1, C_2}$ such that $x \in C_1$ and $y \in C_2$.
\begin{observation}
\label{obs:reduce-to-lca}
    Let $x$ and $y$ be vertices.
    \begin{itemize}
        \item If some compressed subhierarchy $\check{\mathbb{H}}^j \in \check{\mathfrak{H}}$ satisfies the LCA condition for $(x,y)$, then the tree $T_{\mathbb{H}}^j$ preserves the distance between $x$ and $y$ up to a factor $1+\e$.
        \item There is some compressed subhierarchy $\check{\mathbb{H}}^j \in \check{\mathfrak{H}}$ which satisfies the LCA condition for $(x,y)$.
    \end{itemize}
\end{observation}
\begin{proof}
    The first statement follows from \Cref{lem:assigned-good-stretch}: indeed, the LCA of $x$ and $y$ in $\check{\mathbb H}^j$ is some scale-$i$ cluster with $j \equiv i \mod \log_\mu(1/\e)$, and $C$ is assigned to two subclusters that contain $x$ and $y$ respectively, so the assumptions of \Cref{lem:assigned-good-stretch} are satisfied.
    (Note that \Cref{lem:assigned-good-stretch} only claims that distances are preserved up to factor $1+O(\rho \e)$, but this is rescaled to $1+\e$ later in the proof of \Cref{thm:spanning-main}.)

    The second statement follows from the definition of pair-preserving HPF. Indeed, the definition guarantees that there is some hierarchy $\mathbb{H} \in \mathfrak H$ and cluster $C$ in $\mathbb H$ that is assigned to a pair of distinct subclusters $(C_1, C_2)$ with $x \in C_1$ and $y \in C_2$. Let $j$ be the integer in $[1, \log_\mu(1/\e)]$ with $j \equiv i \mod \log_\mu (1/\e)$. By definition of $j$, cluster $C$ appears in the subhierarchy $\mathbb H^j$. Furthermore, cluster $C$ has at least two children subclusters in $\mathbb H^j$, namely $C_1$ and $C_2$; thus $C$ appears in the contracted subhierarchy $\check{\mathbb H}^j$. Finally, observe that $C$ is the LCA of $x$ and $y$ in $\check{\mathbb H}^j$.
\end{proof}

For the rest of this section, we let $\check{\mathbb{H}}^j$ be some fixed compressed subhierarchy; we aim to construct an $O(\log n \cdot \ddim \cdot \log 1/\e)$-bit label on the vertices $V(G)$ such that, given the labels of vertices $x$ and $y$, we can check: \emph{Does $\check{\mathbb{H}}^j$ satisfies the LCA condition for $(x, y)$?} \Cref{obs:reduce-to-lca} implies that constructing these labels is sufficient to prove \Cref{lem:find-correct-tree}. Indeed, we define the complete label of $x \in V(G)$ to be the concatenation of the $O(\log n \cdot \ddim \cdot \log 1/\e)$-bit labels for each of the $\e^{-\tilde O(\ddim)}$ hierarchies $\check{\mathbb{H}}^j$; given the labels of $x$ and $y$, we can find some $\check{\mathbb H}^j$ that satisfies the LCA condition for $(x,y)$ and return the index of tree $T_{\mathbb H}^j$.
The final label has size $\e^{-\tilde O(\ddim)} \cdot \log n$.

\paragraph{LCA labeling.} We begin by reviewing some tools for LCA labeling. Let $T$ be a tree.\footnote{Later, we will take $T$ to be the compressed subhierarchy $\check{\mathbb{H}}^j$; here we state the results for a general tree $T$ to emphasize that we do not need any specific properties of $\check{\mathbb{H}}^j$.} For any two vertices $x$ and $y$ in a tree, let \EMPH{$\lca(x,y)$} denote the lowest common ancestor. For any vertex $x$ in $T$, we define the weight of $x$ to be the number of descendants of $x$ (including itself). A vertex is \EMPH{heavy} if its weight is greater than half the weight of its parent, otherwise it is \EMPH{light}.
Every vertex has at most 1 heavy child.
For every vertex $x$, we define \EMPH{$\apices(x)$} to be the set containing the parent of every light ancestor of $x$. There are $O(\log |V(T)|)$ apices for each vertex $x$, where $|V(T)|$ denotes the number of vertices of $T$. Notice that $\lca(x,y)$ is in $\apices(x) \cup \apices(y)$. The following lemma from \cite{CCL+24b} is implicit from (a small modification of) the LCA labeling scheme of Alstrup, Halvorsen, and Larsen \cite{alstrup2014near}.

\begin{lemma}[Lemma 5.5 of \cite{CCL+24b}, adapted from Corollary 4.17 of \cite{alstrup2014near}]
\label{lem:lca-labeling}
    Suppose we are given, for every vertex $x$ in tree $T$, a function $L_x:V(T) \to \set{0,1}^k$ which assigns each vertex $v \in V(T)$ to some $k$-bit ``name'' $L_x(v)$. (That is, there is a different naming-function on $V(T)$ for every vertex $v$.) Then we can construct a labeling scheme on the vertices $V(T)$ that uses $O(k \log |V(T)|)$ bits per label, such that given the labels of two leaves $x$ and $y$, we can compute
    \begin{itemize}
        \item $L_x(\lca(x,y))$ if $\lca(x,y) \in \apices(x)$
        \item $L_y(\lca(x,y))$ if $\lca(x,y) \in \apices(y)$.
    \end{itemize}
    If $\lca(x,y) \in \apices(x) \cap \apices(y)$, then we can compute the pair $(L_x(\lca(x,y)), L_y(\lca(x,y)))$. This computation can be done in $O(1)$ time.
\end{lemma}

\paragraph{Choosing the names $L_x$.} For every cluster $C$ in $\check {\mathbb H}^j$, assign an arbitrary ordering to its children. There are $\e^{-O(\ddim)}$ children
per cluster, so this ordering lets us specify a child of $C$ using only $O(\ddim \cdot \log 1/\e)$ bits. Let $x \in V(G)$ be a vertex, and treat $x$ as a leaf of $\check {\mathbb H}^j$. For every non-leaf cluster $C$ in $\check{\mathbb H}^j$, we define the ``name'' \EMPH{$L_x(C)$} to be a bit-string consisting of three parts:
\begin{itemize}
    \item (L1) Record which child cluster of $C$ contains $x$, using $O(\ddim \cdot \log 1/\e)$ bits.
    \item (L2) Record which child cluster of $C$ is \emph{heavy} (if there is one), using $O(\ddim \cdot \log 1/\e)$ bits.
    \item (L3) Recall that $C$ is assigned to a pair of children $(C_1, C_2)$ in $\check {\mathbb H}^j$. Record this assignment, using $O(\ddim \cdot \log 1/\e)$ bits.
\end{itemize}
Overall, $L_x(C)$ consists of $O(\ddim \cdot \log 1/\e)$ bits.

\paragraph{Defining and decoding labels.}
Define the label of each vertex in $V(G)$ according to \Cref{lem:lca-labeling}, with the definition of $L_x$ from the previous paragraph. Each label is $O(\log n \cdot \ddim \cdot \log 1/\e)$ bits, as required. We now show that, given the labels of two vertices $x$ and $y$, we can determine whether $\check{\mathbb{H}}^j$ satisfies the LCA condition for $(x,y)$. We use \Cref{lem:lca-labeling} to find information about the LCA of $x$ and $y$, in $O(1)$ time. Letting $C = \lca(x,y)$, we have two cases.
\begin{itemize}
    \item \textbf{Case 1:} $C \in \apices(x) \cap \apices(y)$. Then the labeling of \Cref{lem:lca-labeling} lets us recover $L_x(C)$ and $L_y(C)$. In particular, the (L1) parts of $L_x(C)$ and $L_y(C)$ let us recover the child clusters $C_x$ and $C_y$ containing $x$ and $y$, and the (L3) part lets us recover the pair of children $\set{C_1, C_2}$ that $C$ was assigned to. $\check{\mathbb H}^j$ satisfies the LCA condition for $(x,y)$ if and only if $\set{C_x, C_y} = \set{C_1, C_2}$.
    \item \textbf{Case 2:} Suppose that $C \not \in \apices(x) \cap \apices(y)$. Then, because $C \in \apices(x) \cup \apices(y)$, we may assume WLOG that $C \in \apices(x)$ but $C \not \in \apices(y)$. This means that $y$ belongs to a child of $C$ that is heavy. The labeling of \Cref{lem:lca-labeling} lets us recover $L_x(C)$. The (L1) part of the label lets us recover the child cluster $C_x$ containing $x$, the (L2) part lets us recover the child cluster $C_y$ containing $y$, and the (L3) part lets us recover the pair of children $\set{C_1, C_2}$ that $C$ was assigned to. As above, this lets us determine whether or not $\check {\mathbb{H}}^j$ satisfies the LCA condition for $(x,y)$.
\end{itemize}
The decoding procedure in both Case 1 and Case 2 involves only a linear scan over the labels $L_x(C)$ and $L_y(C)$ to extract the identifiers of children clusters and perform equality checks, so it can be done in $O(d \cdot \log 1/\e)$ time in the word RAM model. We have shown there is an $O(\log n \cdot \ddim \cdot \log 1/\e)$-bit labeling scheme that lets us determine whether $\check{\mathbb{H}}^j$ satisfies the LCA condition for any two vertices. Combined with the discussion after \Cref{obs:reduce-to-lca}, this proves \Cref{lem:find-correct-tree}. We have thus completed the proof of \Cref{thm:routing}.

\paragraph{Acknowledgement.~}
Hsien-Chih Chang and Jonathan Conroy are supported by the NSF CAREER Award No.\ CCF-2443017.
Hung Le and Cuong Than are supported by the NSF CAREER Award No.\ CCF-2237288, the NSF Grant No.\ CCF-2121952 and a Google Research Scholar Award. Shay Solomon is funded by the European Union (ERC, DynOpt, 101043159).  Views and opinions expressed are however those of the author(s) only and do not necessarily reflect those of the European Union or the European Research Council.  Neither the European Union nor the granting authority can be held responsible for them.  Shay Solomon is also supported by a grant from the United States-Israel Binational Science Foundation (BSF), Jerusalem, Israel, and the United States National Science Foundation (NSF).



\small
\bibliographystyle{alphaurl}
\bibliography{main,routing}

\end{document}